\documentclass[journal, twocolumn, 10pt]{IEEEtran}
\normalsize

\usepackage{amsmath, amssymb, multirow, cite , epsfig, graphicx , graphics, color, wrapfig, subfig, comment,bm,array} 
\usepackage{mathrsfs}

\ifx \myver \submitverlatercheck 
      \renewcommand\baselinestretch{0.96}
\fi

\def \mycol{2}

\def \coltwo{2}

\newtheorem{theorem}{Theorem}[section]
\newtheorem{lemma}[theorem]{Lemma}
\newtheorem{proposition}[theorem]{Proposition}

\newenvironment{proof}[1][Proof]{\begin{trivlist}
\item[\hskip \labelsep {\bfseries #1}]}{\end{trivlist}}

\newcommand{\qed}{\nobreak \ifvmode \relax \else
      \ifdim\lastskip<1.5em \hskip-\lastskip
      \hskip1.5em plus0em minus0.5em \fi \nobreak
      \vrule height0.75em width0.5em depth0.25em\fi}

\begin{document}

\title{Channel-Adaptive Packetization Policy for Minimal Latency and Maximal Energy Efficiency}

\author{
   \IEEEauthorblockN{
     Abolfazl Razi\IEEEauthorrefmark{1},
     Fatemeh Afghah \IEEEauthorrefmark{1} and
     Ali Abedi\IEEEauthorrefmark{2}}

   \IEEEauthorblockA{
     \IEEEauthorrefmark{1}Dept. of Electrical Engineering and Computer Science, Northern Arizona University} \\
     \IEEEauthorblockA{
     \IEEEauthorrefmark{2}Dept. of Electrical and Computer Engineering, University of Maine} 
 }
\maketitle
\begin{abstract}

This article considers the problem of delay-optimal bundling of the input symbols into transmit packets in the entry point of a wireless sensor network such that the link delay is minimized under an arbitrary arrival rate and a given channel error rate. The proposed policy exploits the variable packet length feature of contemporary communications protocols in order to minimize the link delay via packet length regularization. This is performed through concrete characterization of the end-to-end link delay for zero-error tolerance system with First Come First Serve (FCFS) queuing discipline and Automatic Repeat Request (ARQ) re-transmission mechanism. The derivations are provided for an uncoded system as well as a coded system with a given bit error rate.

The proposed \textit{packetization} policy provides an optimal packetization interval that minimizes the end-to-end delay for a given channel with certain bit error probability. This algorithm can also be used for near-optimal bundling of input symbols for dynamic channel conditions provided that the channel condition varies slowly over time with respect to symbol arrival rate. This algorithm complements the current network-based delay-optimal routing and scheduling algorithms in order to further reduce the end-to-end delivery time. Moreover, the proposed method is employed to solve the problem of energy efficiency maximization under an average delay constraint by recasting it as a convex optimization problem\footnote{The earlier version of this work was presented in part at the 48th Annual Conference in Information Sciences and Systems (CISS), held on March 19-21, 2014 at Princeton University.}.

\end{abstract}


\IEEEpeerreviewmaketitle

\section{Introduction}

In a variety of wireless applications, the pivotal design objective is to minimize the end-to-end latency or to ensure a predefined delay constraints \cite{delayopt_sch2011, sch_2014_thr_delay, 2014_sch_no_backpressure, delay_analysis_11, dlimit1, thrput_under_delay_2013}. 
For instance, in a variant of sensor networks, the arrival packets are marked outdated if the age of information (the timespan from data generation at sensor nodes until delivery to the processing unit) exceeds a predefined limit. This constraint imposes a hard limit on the system end-to-end delay~\cite{Yates_2012_age}.
Majority of these studies focus on developing optimal routing and scheduling policies performed in higher network layers. A common presumption in such methods is that the per-link throughput and delay parameters of a given transmission technique are physical layer parameters, which are mainly determined by out-of-control channel conditions such as noise and interference levels as well as the local user traffic statistics, therefore can not be improved by varying higher layer parameters \cite{delay_aware_scheduling_2013, delay_analysis_11}. In this work, we provide a time-based packetization policy for the entry point of the network that combines the input symbols (e.g. a sensor measurements in a wireless sensor network) into transmit packets such that the resulting per-link delay is minimized by regularizing packet lengths. This algorithm can be used in combination with the current in-network delay-optimal scheduling policies to further minimize the end-to-end delay.

In a packet-based transmission system, a larger packet size reduces the \textit{packetization} overhead. This overhead may be due to the packet header (e.g. addressing bits, control bits, and CRC codes), channel setup time, or even channel contention period in wireless networks with opportunistic scheduling \cite{overhead3, overhead2, Contention_2010}. The lower overhead translates to a shorter average transmission time for each data symbol. 
On the other hand, longer packets may increase the transmission time by imposing longer packet formation time, since the payload data is not accessible at the destination until a packet is formed at the transmitter and is fully delivered to the destination. Moreover, in a noisy environment with a certain bit error probability, a longer packet size increases delay by elevating re-transmission rate~\cite{plen_throughput_energy_efficiency_98, overhead1,overhead5, optimal_packet_masking_2010}.
In other words, packet length has two contradictory effects on the end-to-end latency per symbol. Addressing this essential trade-off and finding the optimum packet length is a key factor to improve the  communication efficiency, which has been mostly overlooked in the conventional system designs.

Old generation communication protocols such as GSM and ATM permit only fixed-length packets. 
However, the contemporary communication protocols such as IEEE 802.16 and IPv6 allow variable length packets to more efficiently adapt to user traffic demands~\cite{ATM98, IEEE-standards1,Deering:1998:IPV:RFC2460}. In this paper, we exploit this feature to improve the average packet delivery time for wireless networks.

Recently, several attempts have been made to increase communication efficiency by customizing packet lengths based on the channel quality factors. For instance, the idea of local packet length adaptation is introduced in \cite{packet_length_berkley2011} in order to maximize throughput in WLAN channels. An approximate blocking probability is found for general packet length distributions in \cite{Blocking_Prob_Delay_2012}. The impact of packet lengths on other performance metrics such as latency, communication range and energy efficiency is also studied in \cite{plen_delay_throughput_ITWC_2004, sensor_2008_plen, plen_throughput_energy_efficiency_98}.
However, the impact of packet length on the end-to-end data delivery time yet to be comprehensively studied, although it is implicitly addressed by re-transmission probability minimizations~\cite{plen_throughput_1996, plen_throughput_1998, plen_throughput_2004, nelson2008_plen}. Besides, most previously reported approaches for wireless Ad-hoc networks consider the saturated traffic model and aim at optimizing throughput and delay by solely reducing the number of dropped packets while ignoring the queuing dynamics. The saturated traffic model does not cover random traffic in most real-world applications such as web-based applications and Ad-hoc sensor networks. Therefore, we consider a probabilistic traffic model, where a stream of input symbols (e.g. measurements of a data source in a sensor network) are generated according to a Poisson process.

In the proposed packetization policy, we aim at regularizing the packet lengths at the entry point of a wireless sensor network, when the measurement data is bundled into transmit packets and is injected to the network in order to minimize the per-link delay for a given channel error rate and input traffic rate. This is performed by characterizing the impact of the packet length on queuing dynamics, re-transmission rate and consequently the end-to-end delay.
We consider a single-hop communication system with a First Come First Serve (FCFS) scheduler, an unlimited buffer size and an Automatic Repeat Request (ARQ) re-transmission mechanism. Although simple, this model highlights and solves the relevant trade-offs and provides insights for more general systems \cite{PER-ARQ1}. In particular, this model can be deployed at the entry point of sensor networks, where the sensor measurement symbols are formed into transmit packets. Furthermore, it can be integrated with delay-optimal scheduling techniques for multi-hop communications. 
In addition to end-to-end delay analysis, we also characterize the required energy per unit symbol for scenarios with energy efficiency optimization under delay constraint requirements.

A time-based aggregation policy is proposed in \cite{OPTIC-Burst} to optimally combine the Poisson-distributed arrival packet bursts for optical networks application. In contrast, we focus on packet generation from the arrival samples and consider the impact of header size and packet re-transmissions. This work is also closely related to \cite{Ephremides_Shrader_lc1,Ephremides_Shrader_lc2}, where the packets in the buffer are bundled into batches up to a certain number and are transmitted after random linear coding over a noisy channel. The service of a batch is completed when all the containing packets are recovered at the destination. Controlling the number of packets in each linear coding block in \cite{Ephremides_Shrader_lc1} is analogous to defining the number of symbols per packets in our proposed scheme. However, the main differences between this work and \cite{Ephremides_Shrader_lc1,Ephremides_Shrader_lc2} are: i) the packet arrival process being Poisson instead of late arrival Bernoulli distribution; ii) considering packet success rate dependency on the indirectly controlled packet length in our scheme; and iii) studying the effect of packetization overhead in delay analysis.

The rest of this paper is organized as follows.
In section~\ref{sec:sysmodel}, the system model is presented. In section \ref{sec:t}, the problem of channel adaptive packetization policy is formulated and the end-to-end latency is analyzed via characterizing various delay elements for a time-based packetization policy. The delay-optimal packetization criterion is found in section \ref{sec:optimal}. Simulation results are provided in section \ref{sec:simulation}, followed by concluding remarks in section \ref{sec:conclusion}.
	
\section{System Model} \label{sec:sysmodel}

Before elaborating on the system model, the notation style used in this manuscript is summarized. Lowercase boldface letters are used for scalar random variables, the capital letters for constants and fixed parameters, the lowercase letters for variables and realizations of random variables, unless otherwise specified explicitly. Subscripts are used to note symbol and packet indices and postscripts are only used for power operation or as type identifier. 
A summary of the parameters is presented in Table I.  
\begin{table}
\caption{Notation summary.}
\centering
\begin{tabular}{ |p{2cm} | p{1cm}| p{4.cm}| }
  \hline                     
  parameter & unit & description \\  
  \hline                     
  $N$ & bit & number of bits for input symbols\\
  $H$ & bit & number of header bits\\
  $\eta=H/N$ & \textemdash & header load ratio\\
  T & sec & adjustable packetization interval\\
  $S_i$ & \textemdash  & $i^{th}$ input symbol\\
  $X_n$ & \textemdash  & $n^{th}$ transmit packet\\
  $\lambda$ & 1/sec & input symbol arrival rate\\
  $\mathbf{t}_i$ & sec & arrival time for symbol $i$ \\
  $\mathbf{\mathbf{\theta}}_i=\mathbf{t}_i-\mathbf{t}_{i-1}$ & sec & interarrival time for symbol $i$ \\
  $\mathbf{f}_i$ & sec & packet formation delay\\
  $\mathbf{d}_i$ & sec & symbol delivery time\\
  $\mathbf{k}_n$ & \textemdash  & number of symbols in packet $n$\\
  $\mathbf{l}_n$ & bit& number of bits in packet $n$ \\
  $\mathbf{a}_n$ & sec & arrival time for packet $n$ \\
  $\bm{\tau}_n=\mathbf{a}_n-\mathbf{a}_{n-1}$ & sec & inter-arrival time for packet $n$\\
  $\beta$ & \textemdash  & bit error probability\\
  $\alpha=1-\beta$ & \textemdash  & bit success probability\\
  $R$ & bit/sec & channel transmission rate \\
  $P_t$ & Watt & transmission power \\
  ECR & Joule/bit & Energy Consumption Rating\\
  $\mathbf{p}^e_n$ &   packet/sec & packet error probability for packet $n$\\
  $\gamma$ & \textemdash  & Euler constant\\
  $\mathbf{s}_n$ & sec  & service time for packet $n$\\
  $\mathbf{w}_n$ & sec & waiting time for packet $n$\\
  \hline  
\end{tabular}
\end{table}

A sequence of $N$-bit symbols $\{S_i\}_{i=0}^{\infty}$ arrives at the input of transmission system according to a Poisson process with rate $\lambda$. 
The symbols are combined into packets with a constant header size $H$, then scheduled in an infinite length queue with FCFS discipline and transmitted through a wireless channel with bit rate $R$ to the destination, as depicted in Fig.~\ref{fig:sysmodel}. 
\begin{figure}[b]
\centering
\includegraphics[width=1\columnwidth]{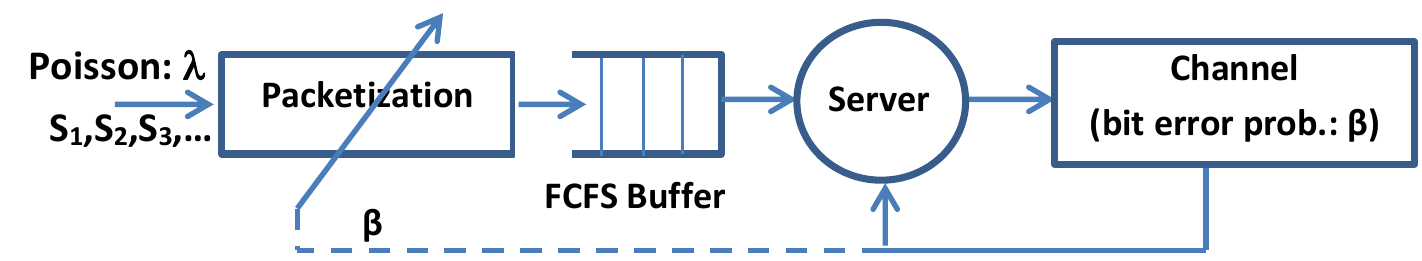}  
\caption{System model: the fixed-length input symbols $\{X_1, X_2, X_3,\dots \}$ with arrival rate of $\lambda$ are combined into transmit packets and are scheduled in a FCFS buffer for transmission over a single-hop uncoded wireless channel with bit error probability $\beta$.}
\label{fig:sysmodel}
\end{figure}

In order to bundle the symbols into the transmit packets, we adopt a \textit{time-based packetization policy}, where the time axis is partitioned into consecutive equal packetization intervals of size $T$. The $\mathbf{k}_n$ symbols that arrive at the $n^{th}$ interval $\big[(n-1)T, nT\big)=\{t|(n-1)T \leq t < nT\}$ are combined to form a single transmission packet $X_n$ and is scheduled for transmission. We propose two different implementation modes. The distinction between these modes is the way we handle the intervals with zero arrival symbols. In \textit{mode 1 (efficient mode)}, no packet is sent for zero symbol accumulation and hence formation of the transmit packet is postponed to the subsequent intervals. Therefore, the inter-arrival for packet $n$, denoted by $\bm{\tau}_n$ can be multiples of packetization interval, $T$. This mode is more efficient and achieves higher channel utilizations in practice. However, in \textit{mode 2 (slotted mode)} a dummy packet of size $H$ is sent for zero symbol accumulation and we have $\bm{\tau}_n=T$. This approach is desired for slotted systems with a constant-length time slot and has the advantage of easy synchronizations and less complex analysis. We represent the two modes with $\mathscr{M}1$ and $\mathscr{M}2$ for the sake of brevity.
In this article, the main focus is on the efficient mode (i.e. $\mathscr{M}1$), but we occasionally mention the difference with the slotted mode (i.e. $\mathscr{M}2$). Both packetization modes are depicted in Fig. \ref{fig:framingTB}. 
\begin{figure}[ht]
\centering
\includegraphics[width=1\columnwidth]{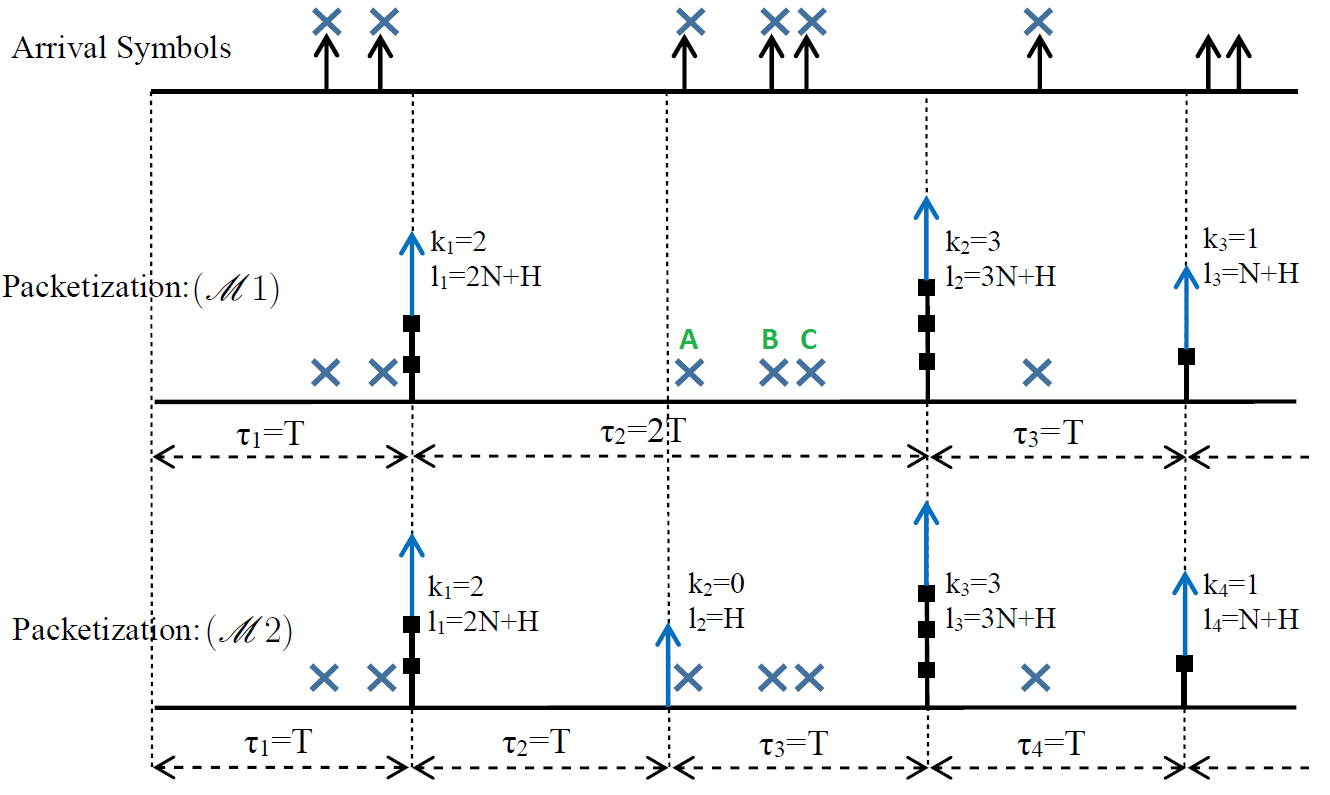} 
\caption{Time-based \textit{packetization policy}: symbols arive at a packetization interval are combined to form transmit packets with constant header sizes. If no symbol is arrived at a packetization interval, packet formation is postponed to the subsequent interval in mode $\mathscr{M}1$, whereas a dummy packet of length $H$ is send in mode $\mathscr{M}2$.}
\label{fig:framingTB}
\end{figure}

The number of bits in the $n^{th}$ packet, $\mathbf{l}_n = H+\mathbf{k}_nN$ is fully determined by the number of encapsulated symbols $\mathbf{k}_n$. For \textit{slotted mode} $\mathscr{M}2$, it is clear that $\mathbf{l}_n$ is Poisson distributed. In \textit{efficient mode} $\mathscr{M}1$, however the packet inter-arrival time $\tau_n$ can span over $\mathbf{m}_n =1,2,3,\dots$ intervals, where all the symbols arrive at the last interval. Therefore, we have
\begin{align} \label{eq:plen}
\nonumber
&\mathbb{P}(\mathbf{l}_n=H+kN)=\\  
&\begin{cases}
\sum_{m_n=1}^{\infty} \mathbb{P}(\mathbf{l}_n=H+kN,\bm{\tau}_n=m_nT)\\ ~~~~~=\sum_{m_n=1}^{\infty}   (P_{\mathbf{k}_n}(0))^{m_n-1}e^{-\lambda T} \frac{(\lambda T)^k}{k!}\\~~~~~= \frac{e^{-\mu} \mu^k}{(1-e^{-\mu} )k!}, ~~~k=1,2,3,\dots ~~~(\mathscr{M}1),\\
\frac{e^{-\mu} \mu^k}{k!}, ~~~~~~~~~~~~~~~~k=0,1,2,\dots ~~~(\mathscr{M}2),
\end{cases}
\end{align}
where $\mu=\lambda T$ is the average number of symbols in an interval of length $T$. The summation in the efficient mode $\mathscr{M}1$ is over various lengths of the packet inter-arrival time.

It is noteworthy that the packet formation time for each packet in the efficient mode $\mathscr{M}1$ is solely defined by the arrival of the first symbol after the preceding packet formation epoch, whereas the packet length is defined by the number of containing symbols or equivalently the number of symbols that arrive after the first symbol and before the end of the current packetization interval. For instance, in Fig. \ref{fig:framingTB}, symbol A arrives in the second interval after packet 1 formation epoch, which delays the packet 2 formation to the end of the second interval with respect to packet 1 formation epoch. However, the length of the second packet is $3N+H$, which is defined by the arrival of symbols B and C after symbol A and before the end of the current packetization interval. Inter-arrival times between the symbols are independent due to the memoryless property of Poisson arrival process. Consequently, the length of packets are independent form their formation time with respect to the preceding packet. This property is also apparent form Eq (\ref{eq:plen}). For slotted mode $\mathscr{M}2$, the independence of packet lengths and the inter-packet times are trivial, since the inter-packets times are fixed. The independence of packet lengths and the inter-packet times simplifies the queuing system analysis through legitimate use of Kingman's formula \cite{Kingman}.

The resulting packets are transmitted over a single-hop uncoded wireless channel with bit error probability of $\beta$. 
For zero error tolerance scenario, the probability of packet error $\mathbf{p}^e_n$ for packet $n$ with length $\mathbf{l}_n$ is a random variable (r.v.) and defined as\footnote{The above equations are derived for uncoded system, where the packets include only information bits. In more sophisticated protocols, different coding schemes for header part and payload data may be used \cite{IEEE802-11-1, IEEE802-11-2}. If the coding rates are $R_D$ and $R_H$ for Data and Header with respective bit error probabilities of $\beta_D$ and $\beta_H$, then one needs to incorporate the following modifications: $\mathbf{l}_n = \mathbf{k}_n  N / R_D + H / R_H$   and $\mathbf{p}^e_n =   1- (1- \beta_H)^H (1-\beta_D)^{\mathbf{k}_n N}$. The rest of equations remain unchanged. In fact, setting $R_D=R_H=1$ and $\beta_H=\beta_D=\beta$, these equations reduce to the derived equations. In the above derivations, we followed the popular approximation of using bit and packet error rates as reliable surrogates for bit and packet error probabilities. Since the above modifications do not change the subsequent analysis, we stick with an uncoded system in the rest of the paper for the sake of notation convenience. However, simulations results are included for a coded system in section \ref{sec:simulation}.}

\begin{align}
\mathbf{p}^e_n = 1-(1-\beta)^{\mathbf{l}_n} = 1-\alpha^{\mathbf{l}_n},
\end{align}
where $\alpha=1-\beta$ is bit success probability used for notation convenience in the subsequent equations. The erroneous packets are successfully detected at the destination using CRC codes and are re-transmitted using ARQ scheme with instantaneous feedback channel until they are successfully delivered to the destination. The number of transmissions for packet $n$, denoted by $\mathbf{r}_n \in \{1,2,3,...\}$ is a Geometrically distributed r.v. and depends on the packet length $\mathbf{l}_n$ and bit success probability $\alpha$ as follows:

\begin{align} \label{eq:arq}
P_{\mathbf{r}_n|\mathbf{l}_n}(r | \mathbf{l}_n=l)=  \alpha^l(1-\alpha ^l )^{r-1}.   
\end{align}

Hence, the expected value of number of re-transmissions can be simply found by calculus of power series as follows:
\begin{align} \label{eq:arq2}
\nonumber
\mathbb{E}_{\mathbf{r}_n|\mathbf{l}_n}[r| \mathbf{l}_n=l] &=\sum_{r=1}^{\infty}r P_{\mathbf{r}_n|\mathbf{l}_n}(r | \mathbf{l}=l) \\&=\sum_{r=1}^{\infty} r \alpha^l (1-\alpha^l) ^{r-1}=\alpha^{-l}.
\end{align}

The $i^{th}$ symbol of the $n^{th}$ packet experiences the end-to-end delay $\mathbf{d}_i$, which includes packet formation delay, waiting time and service time, denoted by $\mathbf{f}_i$, $\mathbf{w}_i$ and $\mathbf{s}_i$, respectively. Packet formation delay $\mathbf{f}_i$ is the time difference between the symbol arrival epoch and the corresponding packet formation epoch. Thus, may be different for symbols inside a packet. Whereas, $\mathbf{w}_i$ and $\mathbf{s}_i$ are packet-based delays and are equal for all symbols in packet $n$ and account for waiting and service times for both primary and retransmit periods. Therefore, we have $\mathbf{w}_i=\mathbf{w}_n \text{ and } \mathbf{s}_i = \mathbf{s}_n$ for all symbols $i$ forming packet $n$. The goal is to find the optimal packetization policy that minimizes the expected average delay based on detail characterization of different delay terms. Under stability conditions, in stationary states the expected delay of packets are equal and so is the expected delay of the samples ($\mathbb{E}[\mathbf{d}]=\mathbb{E}[\mathbf{d}_i]$). Moreover, due to the ergodicity of the queue, we use the time average of the symbols delays obtained from simulations in section \ref{sec:simulation} to compare with analytically derived expression for expected delay $\mathbb{E}[d]$ using the following relation: 

\begin{align} \label{eq:d2}
\mathbb{E}[\mathbf{d}]=\mathbb{E}[\mathbf{d}_i]=\underset{t\rightarrow \infty}{\lim} \frac{1}{M(t)}\sum_{i=1}^{M(t)} d_i,
\end{align}
where $M(t)=\max~\{i:t_i<t\}$ is the number of symbols arrived by time $t$.

\section {Delay Terms in Time Based Framing Policy} \label{sec:t}
In this section, various delay terms for the proposed time-based packetization policy are evaluated.

\subsection{Packet Inter-arrival time} 

The interarrival times between transmit packets are constant in the \textit{slotted mode} and hence we have a deterministic packet generation process $\bm{\tau}_n=T_n-T_{n-1}=T$ with the following moments for interarrival times:
\begin{align} 
\label{eq:ETs}
&\bm{\tau}_n=T  \Longrightarrow ~\mathbb{E}[\bm{\tau}_n]=T, ~\mathbb{E}[\bm{\tau}_n^2]=T^2, ~\sigma^2[{\bm{\tau}_n}]=0 & (\mathscr{M}2).
\end{align} 

However, in the \textit{efficient mode}, the time between two packetization epochs, $\tau_n$ can be extended to multiples of the packetization interval $T$. Note that $\bm{\tau}_n = \mathbf{m}_n T,~\mathbf{m}_n=1,2,3,\dots$ corresponds to zero symbol accumulation at the first $\mathbf{m}_n-1$ intervals followed by an interval with nonzero arrival symbols. Therefore, $\bm{\tau}_n$ is Geometrically distributed with fail parameter $P_0=e^{-\lambda T}$ and we have the following first and second order moments for service time:
\begin{align}
\label{eq:ET}
\nonumber
\mathbb{P}(\bm{\tau}_n&=m_n T) = P_0^{m_n-1}(1-P_0),~~~m_n=1,2,3,\dots\\
\nonumber
&\Longrightarrow \mathbb{E}[\bm{\tau}_n]=\frac{1}{1-P_0}T, ~~~\mathbb{E}[\bm{\tau}_n^2]=\frac{1+P_0}{(1-P_0)^2}T^2,\\
&~~~~~~ \sigma^2[{\bm{\tau}_n}]=\frac{P_0}{(1-P_0)^2}T^2 , ~~~~~~~~~~~(\mathscr{M}1).
\end{align}

\subsection{Service time}

Now, we proceed to characterize service time, denoted by $\mathbf{s}$ for $\mathscr{M}1$. 
For derivation simplicity, we first consider an equivalent system, where at each interval a packet of possibly zero size is formed. The auxiliary service time $\tilde{\mathbf{s}}$ is defined for this equivalent system, such that a virtual dummy packet of length $0$ is generated for zero symbol accumulation during a packetization interval. 
The actual and auxiliary service times are related as follows
\begin{align} \label{eq:ss}
\tilde{\mathbf{s}}_n=\begin{cases} \mathbf{s}_n & \text{for }k \neq 0 \text{ with probability } 1-P_0=1-e^{-\mu} \\
0  & \text{for }k = 0 \text{ with probability } P_0=e^{-\mu}\\
\end{cases}
\end{align}

The length of packets at this equivalent system is $\tilde{\mathbf{l}}_n=h(\mathbf{k}_n)H + \mathbf{k}_nN$, where $h(.)$ is step function defined as 
\begin{align} \label{eq:hk}
h(\mathbf{k}_n) = \begin{cases}
1 & \mathbf{k}_n > 0,\\
0 &\mathbf{k}_n \leq 0.
\end{cases}
\end{align}

Note that $\mathbf{k}_n=M(iT)-M(iT-T)\in \{0,1,2,...\}$ is a nonnegative Poisson distributed integer with mean $\mu$, hence we have
\begin{align} \label{eq:l1}
\nonumber
\mathbb{P}(\tilde{\mathbf{l}}_n = h(k)H + kN )= \mathbb{P}(\mathbf{k}_n &= k)= \frac{e^{-\mu}\mu^k}{k!},\\& k=0,1,2,\dots ~.
\end{align} 

The service time of a packet of length $\tilde{\mathbf{l}}_n$ accounts for sending $\tilde{\mathbf{l}}_n$ bits over the channel, which might be repeated for $\mathbf{r}$ times due to the ARQ re-transmission mechanism.
Noting the fixed transmission rate $R$ bit/sec, re-transmission probability in (\ref{eq:arq}) and packet length probability mass function (pmf) in (\ref{eq:l1}), the following pmf is derived for service time, $\tilde{\mathbf{s}}_n$:

\begin{align}  \label{eq:ps1}
\nonumber
\mathbb{P}\big[\mathbf{\tilde{s}}_n&=\frac{r}{R}.\big(h(k)H+kN\big)\big]\\
\nonumber
&=\mathbb{P}\big[\mathbf{\tilde{s}}_n=\frac{r}{R}.\big(h(k)H+kN\big)|\mathbf{k}_n=k\big] \mathbb{P}[\mathbf{k}_n=k] \\
\nonumber
&=\mathbb{P}\big[\mathbf{r}_n=r|\mathbf{k}_n=k\big] \mathbb{P}[\mathbf{k}_n=k] \\
\nonumber
&=\alpha^{H+kN} \big [1-\alpha^{H+kN}] ^{r-1} e^{-\mu}\frac{\mu^k}{k!},\\
& ~~~~~~~~~~~~~~~ r=1,2,3,...;~ k=0,1,2,...~.
\end{align}

Substituting $h(.)$ in eq (\ref{eq:ps1}), it converts to the following expression:
\begin{align}  \label{eq:fs1}
\mathbb{P}[\mathbf{\tilde{s}}_n]= \begin{cases}
\alpha^{H+kN} \big [1-\alpha^{H+kN}] ^{r-1} e^{-\mu}\frac{\mu^k}{k!}, \\
~~~~~~~~~~~~\text{for } \mathbf{\tilde{s}}_n = \frac{r(H+kN)}{R}, ~~r,k=1,2,\dots,\infty, \\
e^{-\mu}, ~~~~~~ \text{for } \mathbf{\tilde{s}}_n = 0.
\end{cases}
\end{align} 

The countable discrete support set of $\tilde{s}$ is $\tilde{\mathcal{S}}=\{ 0\} \cup \{ r(H+kN)/R: ~r,k=1,2,3,... \}$. Likewise, we have the following pmf of $s$ considering equation (\ref{eq:ss}):
\begin{align}  \label{eq:fs2}
\nonumber
\mathbb{P}[\mathbf{s}_n &= \frac{r(H+kN)}{R}]= 
\frac{\alpha^{H+kN} \big [1-\alpha^{H+kN}] ^{r-1} }{1-e^{-\mu}}e^{-\mu}\frac{\mu^k}{k!},\\ &~~~~~~~~~~~ \text{for }r,k=1,2,\dots,\infty, 
\end{align} 
with the same support set excluding $\{0\}$. A typical pmf of $\mathbf{s}_n$ is depicted in Fig. \ref{fig:ps}. It is noticeable that the pmf presents a comb-like shape and is not a monotonic function. In fact, if we fix the number of encapsulating symbols $\mathbf{k}_n$, the packet length is $\mathbf{l}_n=\mathbf{k}_nN+H$ and the service time conditioned on the number of symbols is geometrically distributed with success parameter $\alpha^{\mathbf{l}_n}$ due to re-transmission occurrence. This is highlighted in red color in Fig. \ref{fig:ps} for $\mathbf{k}_n=3$. On the other hand, for a given number of re-transmissions $\mathbf{r}_n$, service time is proportional to the packet length and hence demonstrates a Poisson-like conditional distribution, which is marked with blue color in Fig.\ref{fig:ps} for $\mathbf{r}_n=2$. The unconditional pmf of $\mathbf{s}_n$ can be viewed as the interlace of a series of Geometrical distributions scaled with Poisson functions (or vice versa). This pmf is used in the sequel to derive the moments of service time as required delay analysis parameters as follows.

\begin{figure}[h]
\centering
\includegraphics[width=1\columnwidth]{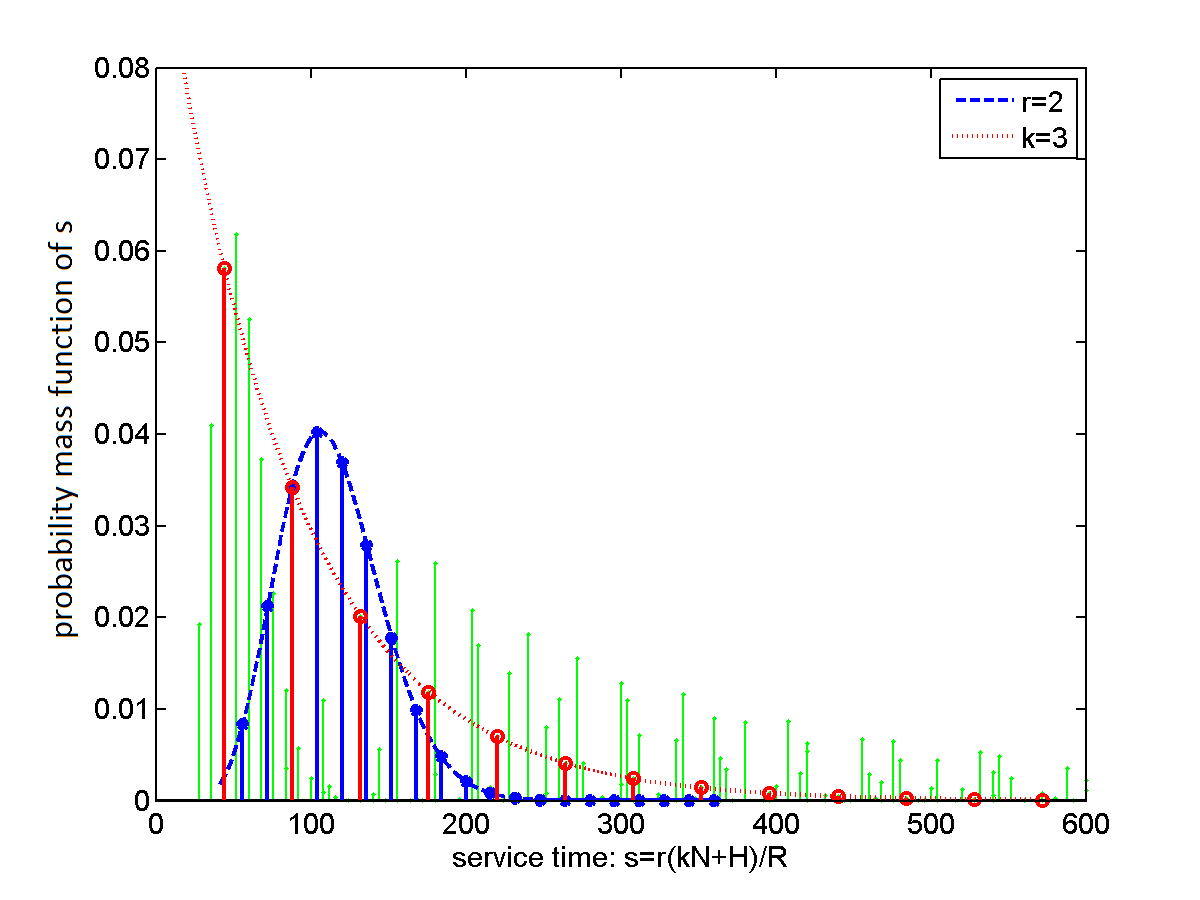} 
\caption{Probability mass function of service time \\$(N=8, H = 20, \lambda = 1, \beta=0.02, R=1)$.}
\label{fig:ps}
\end{figure}

\begin{proposition}  \label{prop:2}
Service time $\mathbf{s}_n$ in \textit{efficient mode} has the following first and second order moments:
\end{proposition}
\begin{align} 
\label{eq:Es1}
&\mathbb{E}[\mathbf{s}_n] =\frac{Ne^{-\mu}}{(1-e^{-\mu})R\alpha^{H}} \Big [  (\eta + \mu{\alpha}^{-N}) e^{ \mu \alpha^{-N} } -\eta  \Big]\\,
\label{eq:Es2}
\nonumber
&\mathbb{E}[\mathbf{s}_n^2]=\frac{N^2 e^{-\mu}}{(1-e^{-\mu}) R^2 \alpha^H}  \Big [ \big ( 2 \mu \alpha^{-2N} + 2 \mu^2 \alpha^{-4N} +\\
\nonumber
& 4\eta \mu \alpha^{-2N} + 2\eta^2  \big) \alpha^{-H} e^{\mu \alpha^{-2N}}  -\big (  \mu \alpha^{-N} +  \mu^2 \alpha^{-2N} +\\
& 2\eta \mu \alpha^{-N} + \eta^2 \big) e^{\mu \alpha^{-N}} +\eta^2 \big(1-2\alpha^{-H}\big)  \Big ], (\mathscr{M}1)
\end{align}
where the header overhead $\eta=H/N$ is defined as the ratio of header size to symbol size.
\begin{proof}
See Appendix.
\end{proof}

The above derivation is obtained for the \textit{efficient mode} based on postponing packet formation to an interval with non-zero arrival symbols. In order to obtain service time distribution for the \textit{slotted mode}, we repeat the same procedure with setting $\tilde{\mathbf{s}}_n=\mathbf{s}_n$ in (\ref{eq:ss}) and replacing the indicator function $h(.)$ with unity function $h(.)=1$ in (\ref{eq:hk}). Consequently, we have $\tilde{\mathbf{l}}_n=\mathbf{l}_n$ and characterizations of $\mathbf{s}_n$ in (\ref{eq:fs2}) encounters the following modification:
\begin{align} 
\nonumber 
\mathbb{P}[\mathbf{s}_n &= \frac{r(H+kN)}{R}] = \alpha^{H+kN} \big [1-\alpha^{H+kN}] ^{r-1} e^{-\mu}\frac{\mu^k}{k!}, \\
& ~~~~\text{ for }  r=1,2,3\dots;~ k=0,1,2\dots
\end{align}

Therefore, the following simplified expressions for the moments of $\mathbf{s}_n$ are obtained after some algebraic manipulations. 
\begin{align} 
\label{eq:Es1-slotted}
&\mathbb{E}[\mathbf{s}_n] =\frac{Ne^{-\mu}}{R\alpha^{H}} \Big [  (\eta + \mu{\alpha}^{-N}) e^{ \mu \alpha^{-N} } \Big], ~~~~(\mathscr{M}2)\\
\label{eq:Es2-slotted}
\nonumber
&\mathbb{E}[\mathbf{s}_n^2]=\frac{N^2 e^{-\mu}}{R^2 \alpha^H}  \Big [ \big ( 2 \mu \alpha^{-2N} + 2 \mu^2 \alpha^{-4N} + 4\eta \mu .\alpha^{-2N} + 2\eta^2 \big) \\
&.\alpha^{-H}  e^{\mu \alpha^{-2N}} -\big (  \mu \alpha^{-N} +  \mu^2 \alpha^{-2N} + 2\eta \mu \alpha^{-N} + \eta^2 \big) e^{\mu \alpha^{-N}}\Big] .
\end{align}
Note that equations (\ref{eq:Es1-slotted},\ref{eq:Es2-slotted}) are equivalent to (\ref{eq:Es1},\ref{eq:Es2}) for large enough $\mu$, since the distinction between two packetization modes is due to intervals with zero arrival symbols whose probability of occurrence diminishes for $\mu \gg 1$.

\textbf{Remark:} There are two extreme cases for packetization interval length, $T$. In the \textit{slotted mode}, when packetization interval is chosen very small, $(\mu \rightarrow 0)$, then approximation $(e^{-\mu}\approx 1-\mu \approx 1)$ implies $\mathbb{E}[\mathbf{s}_n]\approx \frac{(H+N)}{R\alpha^{H+N}}$. 
In this case, the impact of $H$ on the service time is as large as $N$, since most packets include only one symbol and their lengths are $N+H$. On the other hand, for extremely large packetization interval $(\mu \rightarrow \infty)$ and finite $N$ and $H$, we have $N \mu \gg H$. After some mathematical manipulation we have $\mathbb{E}[\mathbf{s}_n] \approx \frac{N \mu}{R \alpha^{N+H}}e^{-\mu (1-\alpha^{-N})}$. In this case, each packet includes a large number of symbols and the impact of header size on the service time is negligible. In particular, for error-free channel, we have $\alpha=1-\beta=1$, hence service time is reduced to $\frac{N \mu}{R}$, which grows linearly with $T$. This is shown in simulation results presented in section \ref{sec:simulation}.

\subsection{Packet formation delay} \label{sec:packet-formation-delay}
To account for the packet formation delay in \textit{efficient mode}, we recall Poisson arrivals of symbols during a framing interval $T$, as depicted in Fig.~\ref{fig:framingTB}. If $\mathbf{k}_n$, $\mathbf{a}_n$ and $\bm{\tau}_n$ are the number of containing symbols, arrival time, and inter-arrival time of the packets, then for all containing symbols $\{S_{n_i} | \bm{a}_{n-1} \leq \mathbf{t}_{n_i} < \bm{a}_{n} \}$, we define the packet formation delay $\mathbf{f}_{n_i}$ as the time span between the symbol arrival epoch $\mathbf{t}_{n_i}$ and the packet formation epoch $\mathbf{a_n}$. The expected value of the average packet formation delay in any packetization interval is:
\begin{align} \label{eq:f1}
\mathbb{E}[\mathbf{f}_{n_i}]= \mathbb{E}_{\mathbf{k}_{n_i}}\big[\mathbb{E}[\mathbf{f}_{n_i}|\mathbf{k}_n~\text{arrivals}]\big]
\end{align}

Noting the memory-less property of exponential distribution, if we choose a packetization window of length $T$ in a random location such that the window encompasses $k_n$ symbols, then the packetization interval includes $k_n+1$ symbol inter-arrival times $\theta_{n_i}$, two of which ($\theta_{n_1}$ and $\theta_{n_k+1}$) are truncated at both ends as depicted in Fig. \ref{fig:interval}. 
Therefore, the expected value of average packet formation delay can be calculated as follows,
\begin{align} \label{eq:f2}
\nonumber
\mathbb{E}[\bar{\mathbf{f}}|\mathbf{k}_n=k] &=  \mathbb{E}\big[ \frac{1}{k} \sum_{i=n_1}^{n_k}  \mathbf{f}_i\big]\\
\nonumber
& =  \mathbb{E}\big[ \frac{1}{k} \sum_{i=n_1}^{n_k} \big(  \frac{1}{2} \theta_{n_k+1} + \sum_{j=i+1}^{n_k} \theta_i \big) \big]\\
\nonumber
&=  \frac{1}{k} \sum_{i=n_1}^{n_k} \big(  \frac{1}{2} \mathbb{E}[\theta_{n_k+1}] + \sum_{j=i+1}^{n_k} \mathbb{E}[\theta_i] \big) \\
\nonumber
&=\frac{1}{k}\big(k\frac{1}{2}+ \frac{k(k-1)}{2}  \big) \frac{1}{\lambda}= \frac{k}{2\lambda} \\
&\Longrightarrow \mathbb{E}[\bar{\mathbf{f}}]=\frac{1}{2\lambda}\mathbb{E}[\mathbf{k}_n] = \frac{1}{2\lambda}.{\lambda T}=\frac{T}{2}
\end{align}
This is consistent with the memoryless and i.i.d properties of inter-arrival distributions. Since it allows us to interchangeably calculate expected distance from either sides of the frame, which in turn implies: $\mathbb{E}[\sum_{i=n_1}^{n_1+k-1} \mathbf{f}_i]  = \mathbb{E}[\sum_{i=n_1}^{n_1+k-1} (\mathbf{t}_i-\mathbf{a}_{n-1})] = kT-\mathbb{E}[\sum_{i=n_1}^{n_1+k-1} \mathbf{f}_i]  \Longrightarrow \mathbb{E}[\sum _{i=n_1}^{n_1+k-1}\mathbf{f}_i]=kT/2\Longrightarrow \mathbb{E}[\mathbf{f}_i]=T/2$ \cite{razi-ciss-packet, stoch-ross-book}.
Consequently, the packet formation delay is linearly proportional to packetization interval as was expected.

\begin{figure}[h]
\centering
\includegraphics[width=1\columnwidth]{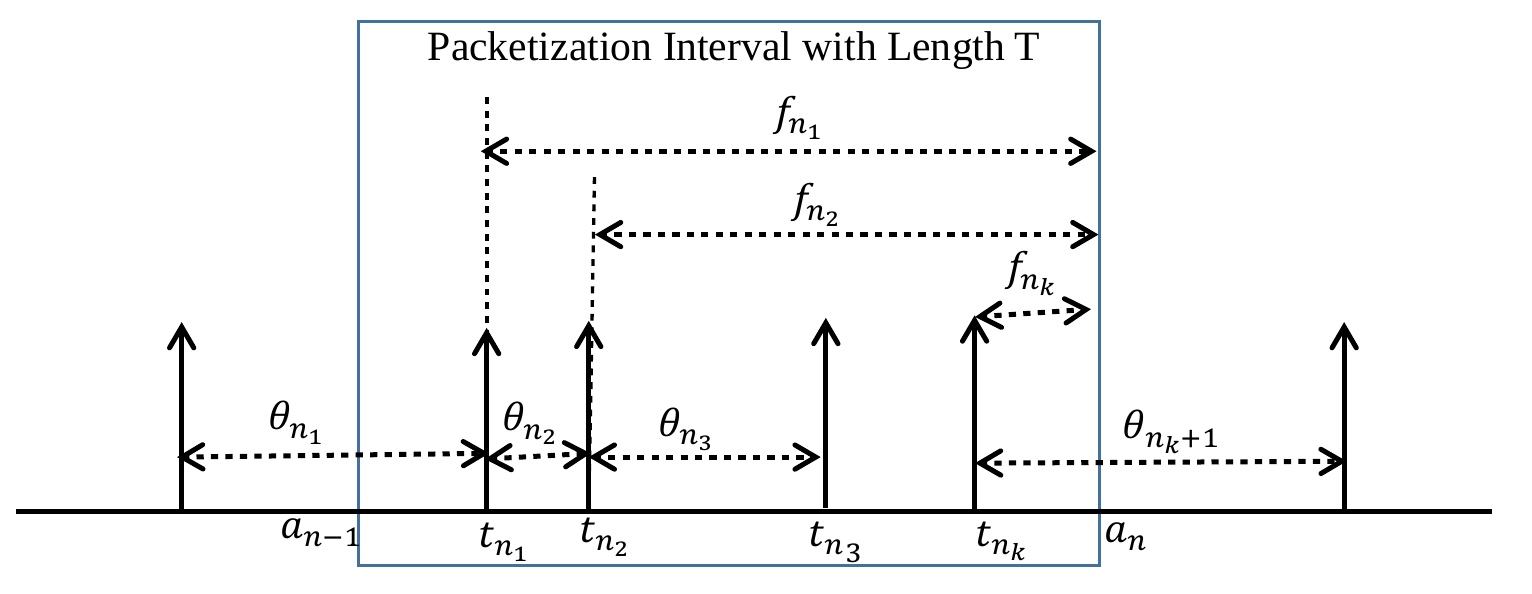}  
\caption{Packetization interval: packet formation delay $\mathbf{f}_{n_i}$ is defined as the time span from the symbol arrival time to the end of current packetization interval $\mathbf{a}_n$.}
\label{fig:interval}
\end{figure}

\subsection{Queuing dynamics} \label{sec:analysis}

For the $n^{th}$ transmit packet $X_n$, the arrival, service and waiting times are presented by $\mathbf{a}_n$, $\mathbf{s}_n$, and $\mathbf{w}_n$, respectively. The inter-arrival time is also denoted by $\bm{\tau}_n=\mathbf{a}_{n}-\mathbf{a}_{n-1}$, as depicted in Fig. \ref{fig:framingTB}. The waiting time for the $n^{th}$packet $\mathbf{w}_{n}$ is a non-negative value and can not exceed the difference between the current inter-arrival time $(\bm{\tau}_{n})$ and the sojourn time of the previous packet $(\mathbf{w}_{n-1} + \mathbf{s}_{n-1})$, hence the following well celebrated Lindley's equation holds \cite{Queu_Book2}:
\begin{align} \label{eq:w1}
\mathbf{w}_{n}= [\mathbf{w}_{n-1} + \mathbf{s}_{n-1} - \bm{\tau}_{n}]^+,
\end{align}
with initial condition $\mathbf{w}_0=0$ and $(x)^+=\max(0,x)$. We already showed in (\ref{eq:plen}) that the packet length and packet inter-arrival times are independent, which implies the independence of service time $\mathbf{s}_n$ and waiting time $\mathbf{\bm{\tau}}_n$. Therefore, we have
\begin{align}
\nonumber
\mathbb{P}(\mathbf{w}_{n}&=w, \bm{\tau}_n \leq t | \mathbf{w}_0,\mathbf{w}_1,...,\mathbf{w}_{n-1},\bm{\tau}_0,\bm{\tau}_1,...\bm{\tau}_{n-1})\\
& =\mathbb{P}(\mathbf{w}_{n}=w, \mathbf{a}_{n+1}-\mathbf{a}_n \leq t | \mathbf{w}_{n-1}).
\end{align}
This means that the process $\{\mathbf{w},\bm{\tau} \} \triangleq \{\mathbf{w}_n,\bm{\tau}_n\}$ forms a renewal Markov Process. To find the transition kernel for the embedded Markov chain $\{\mathbf{w}_n\}$, we note
\begin{align} \label{eq:kernel1}
\nonumber
&\mathbb{P}(\mathbf{w}_{n}= w_{n}| \mathbf{w}_{n-1}=w_{n-1}, \bm{\tau}_n = \tau) =  f_{\mathbf{s}}(w_{n}-w_{n-1}+\tau)u(w_{n})\\
& + \sum_{\substack{s \in \mathcal{S}\\ 0 \leq s \leq \max(0,\mathbf{\tau}-w_n)}} f_{\mathbf{s}}(s)\delta (w_{n}),
\end{align}
where $u(.)$ and $\delta(.)$ are step and Dirac impulse functions and $f_{\mathbf{s}}(s)$ is pmf of serive time, ${\mathbf{s}_n}$ defined in (\ref{eq:ps1}). The first term in (\ref{eq:kernel1}) accounts for the case where packet $n$ arrives before packet $n-1$ service is completed (i.e. $\mathbf{w}_{n-1} + \mathbf{s}_{n-1} - \bm{\tau}_{n} > 0$). Therefore, $\mathbf{w}_{n}=\mathbf{w}_{n-1} + \mathbf{s}_{n-1} - \bm{\tau}_{n}$ is a positive value which occurs for $s_{n-1}=w_{n}-w_{n-1}+\tau$. The second term is corresponding to the case where the packet $n$ arrives after completeion of packet $n$ service and we have $0 \leq \mathbf{s}_{n-1} \leq \max(0,\mathbf{\tau}-w_{n-1})$. In this case, packet $n$ does not experience any waiting and  $\mathbf{w}_{n}=\max(0,\mathbf{w}_{n-1}+\mathbf{s}_{n-1}-\bm{\tau}_{n})$ is mapped to $0$.
Noting that $\bm{\tau}_n$ follows Geometric distribution, (\ref{eq:kernel1}) yields the following transition kernel:

\begin{align} \label{eq:kernel2}
\nonumber
&\mathbb{P}(\mathbf{w}_{n}=w_n| \mathbf{w}_{n-1}= w_{n-1}) \\
\nonumber
&=  \sum_{m=1}^{\infty}\mathbb{P}(\mathbf{w}_{n}=w_n| \mathbf{w}_{n-1}= w_{n-1}, \bm{\tau_n} = m T) \mathbb{P}(\bm{\tau_n} = m T)\\
\nonumber
&= \sum_{m=1}^{\infty}  (1-e^{-\mu})e^{-(m-1) \mu} \Big [   f_{\mathbf{s}}(w_{n}-w_{n-1}+mT)u(w_{n}) \\
&+ \sum_{\substack{s \in \mathcal{S}\\ 0 \leq s \leq \max(0,mT-w_n)}} f_{\mathbf{s}}(s)\delta (w_{n}) \Big ],
\end{align}
This transition kernel is well defined for a given $f_{\mathbf{s}}(s)$ and is used in Monte Carlo simulations in section \ref{sec:simulation}.

\subsection{Stability condition}  \label{sec:stability}
The transition kernel in (\ref{eq:kernel2}) states the evolution of queuing delays until it reaches the stationary state. 
It has been shown that for such a system, the stability is granted if the embedded Markov process complies the sufficient condition of ergodicity defined in \cite{Tweedie, Lindley} as follows
\begin{align} \label{eq:stab1}
\mathbb{E}[\mathbf{s}_n-\bm{\tau}_n]<0.
\end{align}

Therefore, $T$ must be chosen such that $\mathbb{E}[s]<\mathbb{E}[\tau_i]$ is satisfied. Substituting (\ref{eq:ET},\ref{eq:Es1}) in (\ref{eq:stab1}), we have the following stability conditions: 
\begin{align} \label{eq:stab2}
\nonumber
&\frac{  (H + N \mu{\alpha}^{-N}) e^{ -\lambda T (1-\alpha^{-N}) } -H e^{-\lambda T} }{{(1-e^{-\lambda T})R\alpha^{H}}} \leq \frac{T}{1-e^{-\lambda T}}\\
&\Longleftrightarrow R \ge \frac{e^{-\lambda T}\big [  (H + N \lambda T{\alpha}^{-N}) e^{\lambda T \alpha^{-N} } -H  \big] }{\alpha^{H} T}
\end{align}
Equation (\ref{eq:stab2}) provides a closed-form expression for minimum channel rate with stable queue, if the packetization interval $T$ is fixed. It also provides a lower bound on $T$, for a given channel rate, which can be solved numerically or evaluated approximately. For error free channel $\alpha=1$, this simplifies to
\begin{align} \label{eq:stab3}
R \geq  \frac{(H + N \lambda T)  -H e^{-\lambda T} }{T}.
\end{align}
To find the necessary conditions on channel rate $R$, we consider two extreme cases of short and long packetization intervals. If $T\rightarrow \infty$, (\ref{eq:stab3}) yields: $R \geq N \lambda$, which is the rate capable of handling long packets with negligible overhead bits. The other extreme case occurs when $T \rightarrow 0$, where each symbol forms a packet of length $H+N$ upon arrival. In this case, we have $e^{-\mu} \approx 1- \mu$ and (\ref{eq:stab3}) is reduced to $R \geq (N+H) \lambda$. If $ R \in [N \lambda , (N+H) \lambda ]$. There exists a subset of $T$ such that the queue is stable and we can obtain arbitrary near optimal throughput of $\rho = \frac{\mathbb{E}[s]}{\mathbb{E}[\tau]} \rightarrow 1^-$. For $R \leq N \lambda$, there is no $T$ that satisfies the stability conditions and for $R \geq (N+H) \lambda$, all $T$ values produce a stable queue.

\subsection{Expected waiting time}
Equation (\ref{eq:stab1}) ensures that Markov chain $\{\mathbf{w}_n\}$ is recurrent and positive or equivalently the Markov process $\{\mathbf{w}_n, \bm{\tau}_n\}$ is regenerative.
If we set $\mathbf{v}_0=0$ and $\mathbf{v}_j = \sum\limits_{i=n-j+1}^{n}(\mathbf{s}_i-\bm{\tau}_i)$, by recurrence, we can rewrite (\ref{eq:w1}) as follows \cite{Queu_Book, Queu_Book2}
\begin{align} \label{eq:w2}
\mathbf{w}_{n+1}= \max(\mathbf{v}_0,\mathbf{v}_1,...,\mathbf{v}_n).
\end{align}

One may interpret $\mathbf{w}_n$ in (\ref{eq:w1}) as the maximum of accumulated backward steps for a random walk process back to an arbitrary time between $0$ and $n$.
For a stable queue satisfying (\ref{eq:stab1}), we have $\lim_{n \rightarrow \infty} \mathbb{E}[\mathbf{v}_n]=n\mathbb{E}[\mathbf{s}-\bm{\tau}] < \infty $. Hence, $\mathbf{w}_n$ tends to a r.v. $\mathbf{w}=\underset{i}{\sup} ~\mathbf{v}_i$, as $n$ approaches infinity. Therefore, (\ref{eq:w1}) yields
\begin{align} \label{eq:w3}
\mathbb{E}[\mathbf{w}]&= \mathbb{E}[(w+s - \tau)^+],\\
\sigma^2[\mathbf{w}]&= \sigma^2[(w+s-\tau)^+].
\end{align}

The transition kernel in (\ref{eq:kernel2}) does not provide a closed form equation for expected waiting time, $\mathbb{E}[\mathbf{w}]$, but sampling methods such as Markov-Chain Monte Carlo (MCMC) can be used for numerical evaluation of $\mathbb{E}[\mathbf{w}]$ for given parameters. Here, we use the well-celebrated approximate method of Kingman's formula for G/G/1 queues to obtain the closed-form expression for $\mathbb{E}[\mathbf{w}]$ \cite{Kingman}. Noting the inter-arrival time distribution in (\ref{eq:ET}), we have
\begin{align} \label{eq:w4}
\nonumber
\mathbb{E}[\mathbf{w}] &\approx \frac{\rho \mathbb{E}[\mathbf{s}](C^2[\mathbf{s}]+C^2[{\bm{\tau}}])}{2(1-\rho)} \\
& \approx \frac{\big(\mathbb{E}[\mathbf{s}]\big)^2(1-e^{-\mu})}{T-\mathbb{E}[\mathbf{s}](1-e^{-\mu})} .\frac{C[\mathbf{s}]^2+C[\bm{\tau}]^2}{2} 
\end{align}
where $C[x]=\sigma[x]/\mathbb{E}[x]$ is the coefficient of variation of r.v. $x$ and $\rho$ is the utilization factor of queue defined as $\rho= \mathbb{E}[\mathbf{s}]/\mathbb{E}[\bm{\tau}]= \mathbb{E}[\mathbf{s}](1-e^{-\mu})/T$. 

Substituting (\ref{eq:Es1}, \ref{eq:Es2}) in (\ref{eq:w4}) and noting that $C^2[\bm{\tau}]=(\sigma[\bm{\tau}]/\mathbb{E}[\bm{\tau}])^2=e^{-\lambda T}$, equation (\ref{eq:w4}) provides a closed form equation for waiting time. The simulation results provided in section \ref{sec:simulation} confirm the accuracy of the proposed analysis.
 
\section{Optimal Packetization Interval} \label{sec:optimal}
In stationary situations, as mentioned in section \ref{sec:t}, $\mathbf{w}_i \rightarrow \mathbf{w}$. Noting stationary property of $\mathbf{s}_i$ and $\mathbf{f}_i$, the overall delay term $\mathbf{d}_i$ also approaches a r.v. $\mathbf{d}=\mathbf{s}+\mathbf{w}+\mathbf{f}$.
Substituting (\ref{eq:w3}, \ref{eq:f2}) in (\ref{eq:d2}), results in
\ifx \mycol \coltwo
\begin{align} \label{eq:Ed3}
\nonumber
\mathbb{E}&[d]=\underset{t\rightarrow \infty}{\lim} \frac{1}{M(t)}\sum_{i=1}^{M(t)} \mathbb{E}[d_i]= \mathbb{E}[d_i]\\
\nonumber
&=\mathbb{E}[w_i]+\mathbb{E}[s_i]+\mathbb{E}[f_i] \\
& \approx \frac{\big(\mathbb{E}[s]\big)^2(1-e^{-\mu})}{T-\mathbb{E}[s](1-e^{-\mu})} .\frac{C_s^2+C_{\tau}^2}{2} + \mathbb{E}[s] + T/2.
\end{align}
\else
\begin{align} \label{eq:Ed3}
\nonumber
\mathbb{E}[d]&=\mathbb{E}[d_i] =\mathbb{E}[w_i]+\mathbb{E}[s_i]+\mathbb{E}[f_i] \\
& \approx \frac{\big(\mathbb{E}[s]\big)^2(1-e^{-\mu})}{T-\mathbb{E}[s](1-e^{-\mu})} .\frac{C_s^2+C_{\tau}^2}{2} + \mathbb{E}[s] + T/2.
\end{align}
\fi
Substituting moments of service time and inter-arrival times defined in (\ref{eq:ET},  \ref{eq:Es1},\ref{eq:Es2}) in (\ref{eq:Ed3}) provides a closed form expression for the end-to-end delay in terms of $(N,H,\beta,T)$. 
The convexity of (\ref{eq:Ed3}) with respect to $T$ can easily be verified by checking positivity of the second derivative, which is straightforward but involves many terms. This was also confirmed by numerical evaluation as depicted in Fig.~\ref{fig:delay_per}. Equating derivative of (\ref{eq:Ed3}) with respect to $T$ to zero $(\frac{\partial \mathbb{E}[\mathbf{d}]}{\partial T}=0)$ provides the optimum packetization interval $T^\ast$, which minimizes the end-to-end latency $\mathbb{E}[\mathbf{d}]$. Note that convexity is not a necessary requirement, since the closed form expression enables finding the global minimum using numerical methods such as gradient descent algorithm. We discuss some special cases next.

\subsection{Error-free channel}
The expression in (\ref{eq:Ed3}) is a complex expression in general. 
However, for some reasonable assumptions, it can be further simplified. For instance, we analyze the system for an almost error-free channel, where $\alpha=1-\beta \rightarrow 1$. In order to analyze the delay variations, we consider two extreme cases for $\mu$. For $\mu=\lambda T \gg 1$, we use the approximations $e^{-\mu}\approx 0$ and $1-e^{-\mu}\approx 1$. In this case, we can use approximation $\alpha^{-kN} =\frac{1}{\alpha^{kN}} \approx \frac{1}{1-kN(1-\alpha)} \approx 1+kN\beta$, which arises from Taylor expansions of $\alpha^{-kN}$ around $\alpha=1$. Using these approximations and keeping dominant terms, the moments of service time derived in (\ref{eq:Es1},\ref{eq:Es2}) for $(\alpha \rightarrow 1, \mu \rightarrow \infty)$ are simplified to:
\begin{align} 
\label{eq:Es_app1}
\nonumber
\mathbb{E}[\mathbf{s}] &=\frac{Ne^{-\mu}}{(1-e^{-\mu})R\alpha^{H}} \Big [  (\eta + \mu{\alpha}^{-N}) e^{ \mu \alpha^{-N} } -\eta  \Big] \\
\nonumber
&\approx \frac{N\mu e^{-\mu(1-\alpha^{-N})}}{R\alpha^{H+N}}\approx \frac{N\mu e^{\mu N \beta}}{R\alpha^{H+N}},\\
\nonumber
\mathbb{E}[\mathbf{s}^2]&\approx \frac{N^2 e^{-\mu}}{(1-e^{-\mu}) R^2 \alpha^H}   \big (2 \mu^2 \alpha^{-4N} \big) \alpha^{-H} e^{\mu \alpha^{-2N}}   \\
& \approx \frac{2 \mu^2 N^2 e^{-\mu(1- \alpha^{-2N})}}{ R^2 \alpha^{2H+4N}}  \approx \frac{2 \mu^2 N^2 e^{2\mu N \beta }}{ R^2 \alpha^{2H+4N}}      \\
\nonumber
\Longrightarrow &C^2[\mathbf{s}]=\frac{\mathbb{E}[\mathbf{s}^2]-\big(\mathbb{E}[\mathbf{s}]\big)^2}{\big(\mathbb{E}[\mathbf{s}]\big)^2} 
\approx 2 \alpha^{-2N}-1 \approx 1+4N\beta
\end{align}

Obviously, $\mathbb{E}[\mathbf{s}]$ and $\mathbb{E}[\mathbf{s}^2]$ are convex increasing functions of $\mu$.
Substituting the above approximations  in (\ref{eq:Ed3}) and noting $C[\bm{\tau}^2]\approx P_0\approx 0$ for $\mu \gg 1$, we obtain the following closed form expression for $\mathbb{E}[\mathbf{d}]$:
 \begin{align} \label{eq:Ed4} 
\nonumber
\mathbb{E}[\mathbf{d}]&\approx \frac{\big(\mathbb{E}[\mathbf{s}]\big)^2 C^2[\mathbf{s}]}{2(T-\mathbb{E}[\mathbf{s}])} + \mathbb{E}[\mathbf{s}] + T/2 \\
&\approx  \frac{N\mu e^{\mu N \beta}}{R\alpha^{H+N}} \big[ \frac{(1+4N\beta) N\mu e^{\mu N \beta}}{2TR\alpha^{H+N}-2N\mu e^{\mu N \beta}} +1\big]+T/2 
\end{align}
which is also a convex and increasing function of $\mu$ by simple inspection, since $f(g(x))=\frac{g(x)}{T-g(x)}$ is an increasing convex function of $x$ for positive increasing convex function $g(x)$, provided that the denominator remains positive \cite{boyd-opt}. The positivity of the denominator is ensured by stability conditions of the queue. In this case, $\mathbb{E}[\mathbf{d}]$ grows exponentially with $\mu N \beta$.

For the other extreme case of $\mu=\lambda T \rightarrow 0$ for error free channel ($\alpha \rightarrow 1$), we can use approximation $e^{-\lambda T} \approx 1-\lambda T \approx 1$ and $e^{\lambda T \alpha^{-kN}} \approx 1+\lambda T \alpha^{-kN} \approx 1+\lambda T(1+k N\beta)$. Hence, the moments of $\mathbf{s}$ after removing second order terms of $\mu$ and $\beta$ become
\begin{align} 
\label{eq:Es_app2}
\nonumber
\mathbb{E}[\mathbf{s}] &\approx \frac{(N+H)e^{-\mu}}{ R\alpha^{H}}, ~~~\mathbb{E}[\mathbf{s}^2] \approx  (H+N)^2\frac{e^{-\mu}}{R^2\alpha^{H}}\\
&\Longrightarrow C^2[\mathbf{s}] \approx  e^{\mu}\alpha^{H}-1\longrightarrow 0  \end{align}

This result is consistent with the fact that for extremely small $\lambda T$, the frames include only one packet of length  $(N+H)$ and noting $\alpha \rightarrow 1$, we expect $\mathbb{E}[\mathbf{s}^k]\rightarrow (\frac{N+H}{R})^k$ with negligible variations. Moreover, the inter-arrival time between packets tends to the inter-arrival time between the input symbols and follows an exponential distribution, which is simply verified by $\mathbb{E}[\bm{\tau}]=\frac{T}{1-e^{-\lambda T}}\approx \frac{1}{\lambda}$ in (\ref{eq:ET}). Similarly, (\ref{eq:ET}) yields $C^2[\bm{\tau}] = \frac{\mathbb{E}[\bm{\tau}]^2]}{\mathbb{E}[\bm{\tau}]]^2}-1=e^{-\lambda T}\approx 1$.
Substituting the aforementioned approximations along with the approximation $\mathbb{E}[\mathbf{f}]=\frac{T}{2}(1+e^{-\lambda T}) \approx T$ for $\mu \rightarrow 0$, we obtain the following simplified expression for $\mathbb{E}[\mathbf{d}]$:
\begin{align} \label{eq:Ed-small-T}
\nonumber
\mathbb{E}[\mathbf{d}]&\approx \mathbb{E}[\mathbf{s}] [\frac{\lambda \mathbb{E}[\mathbf{s}]}{1-\lambda \mathbb{E}[\mathbf{s}]}\frac{e^{-\lambda T}}{2}+1]+T/2\\
&\approx  \frac{(N+H)e^{-\mu}}{ R\alpha^{H}}[\frac{(N+H)e^{-\mu}}{2(R\alpha^{H}-\lambda(N+H)e^{-\mu})}+1]+T 
\end{align}
In this case, $\mathbb{E}[\mathbf{d}]$ is also a convex function of $\mu=\lambda T$ by similar argument. These approximations can be used as an accurate estimate of expected value of end-to-end delay for two extreme cases. It also verifies the convexity of $\mathbb{E}[\mathbf{d}]$ with respect to $T$ in general as confirmed by simulations results in section \ref{sec:simulation}. The accuracy of these approximations for extreme cases of extremely large and small $T$ is shown in Fig.\ref{fig:approx}.

\begin{figure}[h]
\centering
\includegraphics[width=1\columnwidth]{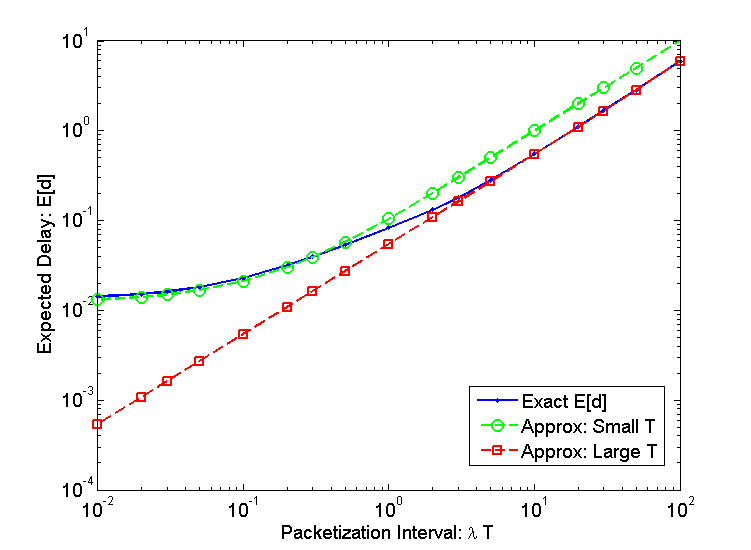} 
\caption{Expected delay $\mathbb{E}[d]$ vs packetization interval $\lambda T$: approximations for extremely large and small packetization intervals $(N=8, H = 16, \lambda = 10, \beta=10^{-3}, R=2\times 10^3)$.}
\label{fig:approx}
\end{figure}

The approximation in (\ref{eq:Ed-small-T}) is based on the assumption that the denominator remains positive for arbitrary low interval. This condition holds for $R$ large enough that satisfies stability condition $R>(N+H)\lambda$ as stated in section \ref{sec:stability}. 
However, for $R<(N+H)\lambda$, we are not able to make $\lambda T$ arbitrary small and $\mathbb{E}[d]$ approaches infinity when $\rho \rightarrow 1^{-}$. In this case $\mathbb{E}[d]$ has a cup shape as depicted in Figs. \ref{fig:delay_sim} and \ref{fig:delay_per}.

Due to the convexity of $\mathbb{E}[\mathbf{d}]$ with respect to $T$, there exists a global minimum corresponding to an optimal packetization interval length with minimum end-to-end latency. We note that the convexity is not necessary and the weaker condition of quasi-convexity is sufficient to ensure uniqueness of the global minimum, which in this case can be found using numerical methods (e.g. gradient descent algorithm).

\subsection{Energy efficiency}
In the above formulations, the optimization is aimed at minimizing the expected delay of symbols. However, one may be interested in maximizing other performance metrics under certain average delay constraints (i.e. $\mathbb{E}[d] \leq D_0$ ). 
Since the derived closed-form expression for end-to-end latency is convex with respect to $T$, optimizing other parameters can be recast as a convex optimization problem given that the objective function is convex with respect to $T$. 

An important parameter that recently regained attention is energy efficiency due to its crucial role in the network maintenance cost and nodes operational lifetime \cite{efa2}. 
We consider a scenario that one is interested in maximizing energy efficiency by choosing appropriate packetization time, when a certain expected average delay is tolerable. In this work, we follow the popular definition of energy efficiency as the ratio of throughput to transmit power which is equivalent to the number of bits transmitted per unit energy in $bit/Joule$~\cite{EE2,EE4,efa2}. In order to cast this objective as a standard convex problem, we formulate it as the equivalent problem of minimizing the energy consumption Rating (ECR) defined as expected energy used for one information bit transmission \cite{EE1}. In the proposed model, $\mathbf{r}(\mathbf{k}N+H)$ bits are sent over the channel for $\mathbf{k}N$ information bits, where the number of symbols in packet, $\mathbf{k}$ and the number of re-transmissions, $\mathbf{r}$ are random variables. Considering channel rate $R$ and transmit power $P_t$, and using the ergodicity of the queue under stability conditions, the ECR is simply $\mathbb{E}[\frac{\mathbf{r}(\mathbf{k}N+H)}{R P_t\mathbf{k}N}]$.
Here, $\mathbf{k}$ is a $T$-dependent random variable with Poisson-like distribution excluding zeros: $\mathbb{P}(\mathbf{k}=k)=\frac{1}{1-e^{-\lambda T}}e^{-\lambda T} (\lambda T)^k/k!$ for $k=1,2,\dots$. Likewise, $\mathbf{r}$ is Geometrically distributed with success parameter $\alpha^{H+\mathbf{k}N}$ and hence depends on $T$ through $k$. To achieve maximum energy efficiency, we desire to minimize ECR($T$). Let us state the following lemma.

\begin{lemma} \label{lem:Ef}
If $\mathbf{k}$ is a Poisson r.v. after zero deletion (i.e.  $\mathbb{P}(\mathbf{k}=k)=\frac{1}{1-e^{-\mu}}e^{-\mu} \mu^k/k!$), then we have
\begin{align} \label{eq:Eflemma}
\mathbb{E}_\mathbf{k}[\frac{1}{\mathbf{k} \xi^\mathbf{k}}]=\frac{\mathcal{E}(\mu/\xi)- \log(\mu/\xi)-\gamma}{e^\mu-1}
\end{align}
where $\log(.)$ is the natural logarithm, $\gamma$ is the \textit{Euler constant} and $\mathcal{E}(.)$ is the \textit{exponential integral} defined in \textit{Cauchy principal value} form as follows:
\begin{align} \label{eq:Ei}
\mathcal{E}(x)=\int_{-\infty}^{x} \frac{e^{t}}{t}\text{dt}
\end{align}
\end{lemma}
\begin{proof}
See Appendix.
\end{proof}

Now we proceed with evaluating $E_f$ in terms of $T$ as follows:
\begin{align} \label{eq:Ef}
\nonumber
RP_t&\text{ECR}(T)  =\mathbb{E}_\mathbf{k}\big[\mathbb{E}_\mathbf{r}[\frac{\mathbf{r}(\mathbf{k}N+H)}{\mathbf{k}N}]\big]\\
\nonumber
&= \mathbb{E}_\mathbf{k}[\alpha^{-(\mathbf{k}N+H)} .\frac{\mathbf{k}N+H}{\mathbf{k}N}]\\
\nonumber
&= \frac{1}{\alpha^H}\mathbb{E}_\mathbf{k}[\alpha^{-\mathbf{k}N}] + \frac{H}{N\alpha^H}\mathbb{E}_\mathbf{k}[\frac{\alpha^{-\mathbf{k}N}}{\mathbf{k}}]\\
\nonumber
&\overset{(a)}=\frac{e^{-\mu}}{\alpha^H(1-e^{-\mu})}(e^{\mu\alpha^{-N}}-1) + \frac{H}{N\alpha^H}\mathbb{E}_\mathbf{k}[\frac{\alpha^{-\mathbf{k}N}}{\mathbf{k}}]\\
\nonumber
&\stackrel{(b)}=\frac{e^{-\mu}}{\alpha^H(1-e^{-\mu})}(e^{\mu\alpha^{-N}}-1) \\
\nonumber
&~~~+ \frac{H}{N\alpha^H (e^{\mu}-1)}\big(\mathcal{E}(\mu \alpha^{-N}) - \log(\mu \alpha^{-N})\big)\\
&=\frac{e^{\mu\alpha^{-N}}-1 + \eta \big(\mathcal{E}(\mu \alpha^{-N}) - \log(\mu \alpha^{-N})-\gamma \big)}{\alpha^H(e^{\mu}-1)}
\end{align}
where $(a)$  and $(b)$ are respectively due to lemmas (\ref{lem:2}) and (\ref{lem:Ef}). If no header bits are used in the system $\eta=H/N=0$, then (\ref{eq:Ef}) simplifies to
\begin{align} \label{eq:Ef2}
\text{ECR}(T)=\frac{e^{\mu\alpha^{-N}}-1}{e^{\mu}-1}RP_t
\end{align}

The average energy required per unit bit transmission in (\ref{eq:Ef}) provides an explicit relation between the packetization interval and the energy efficiency. This relation is depicted for a typical system parameter set in Fig. \ref{fig:EF}. To achieve maximal energy efficiency under average delay constraint, the formal problem formulations becomes:
\begin{align}
\nonumber
&
\begin{cases}
&T^{\ast}=\underset{T}{\text{arg min} }\{\text{ECR}(T)\}\\
&\text{s.t.} ~~\mathbb{E}[\mathbf{d}(T)] < D_0
\end{cases}
 \Longrightarrow  \\
&
\begin{cases}
&T^{\ast}={\text{arg min}} \{  \frac{e^{\mu\alpha^{-N}}-1 + \eta \big(E_i(\mu \alpha^{-N}) - \log(\mu \alpha^{-N})-\gamma \big)}{\alpha^H(e^{\mu}-1)}\}\\
&\text{s.t.} ~~\frac{\big(\mathbb{E}[\mathbf{s}]\big)^2(1-e^{-\mu})}{T-\mathbb{E}[\mathbf{s}](1-e^{-\mu})} .\frac{C^2[\mathbf{s}]+C^2[\bm{\tau}]}{2} + \mathbb{E}[s] + T/2< D_0
\end{cases}
\end{align}
which can be readily solved using optimization packages such as CVX \cite{boyd-opt}. Since the optimization parameter is only packetization time $T$, a short-cut solution is to check K.K.T conditions, namely choosing the best point among three points including the global minimum of ECR($T$) and two solutions for $\mathbb{E}[\mathbf{d}]=D_0$. 
If the delay constraint is too loose, the global minimum of ECR($T$) lays in the valid range of $\mathbb{E}[\mathbf{d}]<D_0$ and the maximal energy efficiency is achieved. Otherwise, one of the extreme corner points that are obtained by solving the constraint equation ($\mathbb{E}[\mathbf{d}]=D_0$) defines the optimal feasible energy efficient point. 
Both scenarios are depicted in Fig. \ref{fig:EF} within section \ref{sec:simulation}.

\section {Simulation results} \label{sec:simulation}
\begin{figure}[b]
\centering
\includegraphics[width=1\columnwidth]{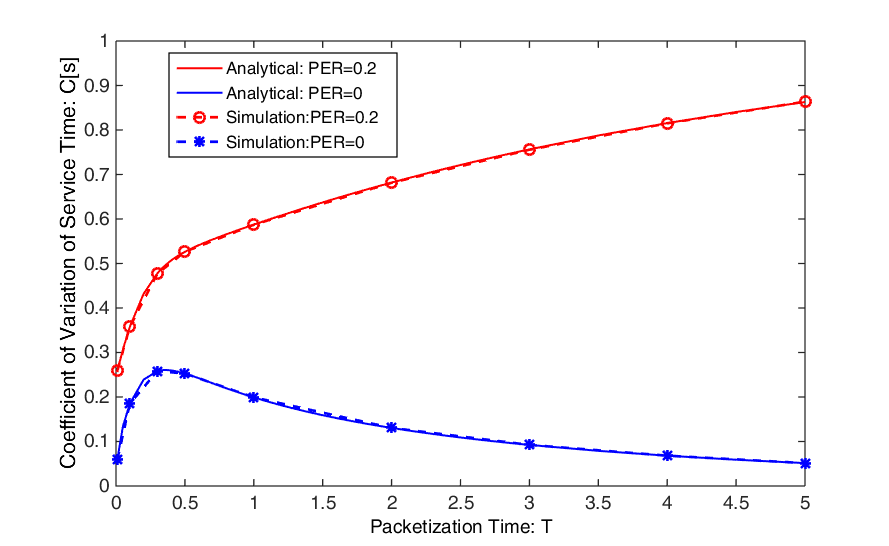}  
\caption{Coefficient of variation of service time $C[\mathbf{s}]=\sigma[\mathbf{s}]/\mathbb{E}[\mathbf{s}]$: comparison between simulations and analysis $(N=16, H = 40, \lambda = 10)$.}
\label{fig:Ks}
\end{figure}

In this section, simulation results are provided to confirm the accuracy of the derived delay optimal packetization criterion. For each graph, we performed Monte-Carlo simulation for at least $100,000$ packets and used the sample average in stationary states as surrogates for the expected values using the ergodicity of the queuing system. The simulation parameters are arbitrarily set to $\lambda = 10, N=16, H=40$ unless otherwise specified. Fig.~\ref{fig:Ks} presents the derived coefficient of variation for service time $C[\mathbf{s}]={\sigma_s}/{\mathbb{E}[s]}$. A perfect match between the analytically derived $C[\mathbf{s}]$ with the empirical results where mean and variance are obtained by time averaging, verifies the accuracy of equations (\ref{eq:Es1}), (\ref{eq:Es2}) and (\ref{eq:Es1-2}).  It is seen that $C[\mathbf{s}]$ increases as packetization time $T$ moves away from zero. This variation is due to the fact that for $T\rightarrow 0$, the packet length tends to include only one symbol and presents a fixed length of $N+H$. Hence, the service time which is proportional to the packet length presents low variations. For moderate $T$ values $\mu \approx 1$, the packet length presents more unpredictability. Furthermore, when $T$ grows to infinity, due to the accumulation of Poisson arrivals over long interval, the number of symbols tends to $\mathbb{E}[k]T=\lambda T$ due to the Law of large numbers. Consequently the packet lengths approach $\lambda T N+ H$ and present low variations. Therefore, coefficient of variation approaches zero for error free channels, $\text{PER}=0$. However, for an erroneous channel, larger packet lengths experience higher packet drop rates. In this case, service time encounters large variations due to variation in the number of packet re-transmissions. The peak of this graph for error free channels is corresponding to a packetization time that results in the most unpredictable service time.

\begin{figure}[t]
\centering
\includegraphics[width=1\columnwidth]{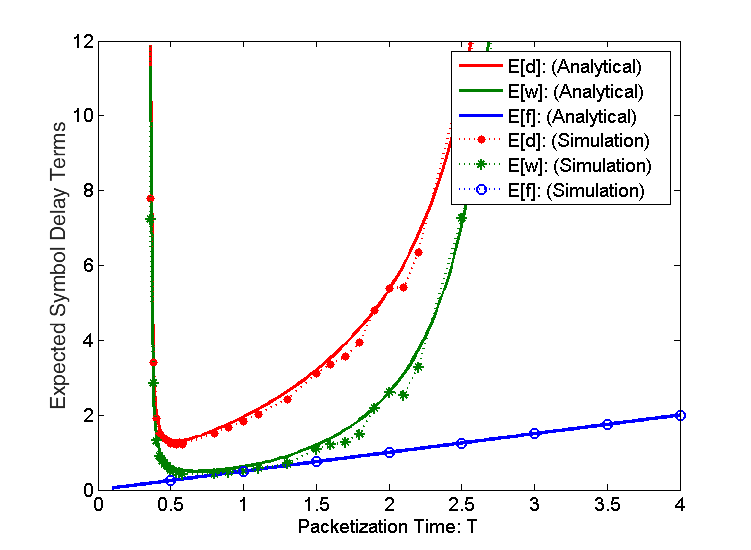}  
\caption{Various delay terms vs packetization interval $(N=16, H = 30, \lambda = 10, \beta=10^{-3},R=300, )$.}
\label{fig:delay_sim}
\end{figure}

Fig.~\ref{fig:delay_sim} presents the impact of packetization interval length ($T$) on various delay sources. Solid lines in this figure represent analytically derived delays, while dashed lines with markers represent the empirical values (sample means) obtained by Monte-Carlo simulations. The packet formation delay, which is shown with blue curve is proportionally related to $T$ as was intuitively expected and is discussed in section \ref{sec:packet-formation-delay}. 
It is noticeable that the expected waiting time $\mathbb{E}[\mathbf{w}]$ is a convex function of $T$, and so is the average end-to-end delay $\mathbb{E}[\mathbf{d}]$. Convexity of delay curves with respect to $T$ guarantees the uniqueness of the optimum packetization interval length for a given set of system parameters. These results suggest that by choosing an appropriate packetization interval, we can minimize the average end-to-end delay.

\begin{figure}[h]
\centering
\includegraphics[width=1\columnwidth]{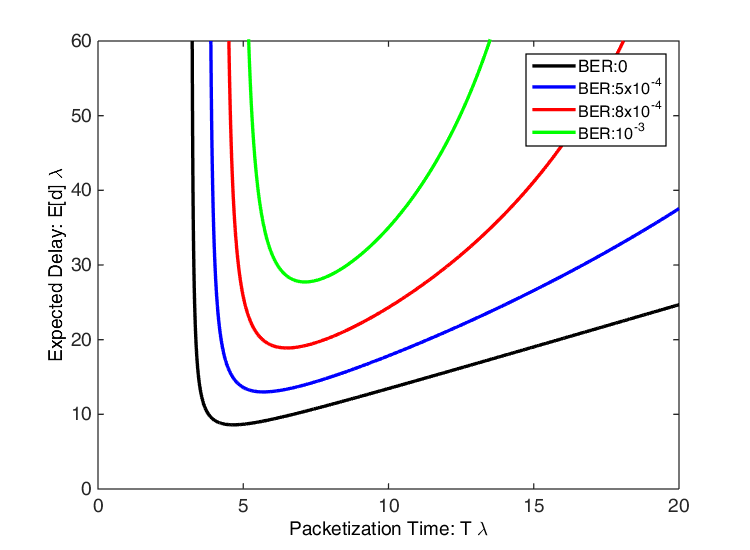}  
\caption{Expected delay $\mathbb{E}[\mathbf{d}]$ vs packetization interval $\lambda T$ for different channel error probabilities $(N=16, H = 30, \lambda = 10)$.}
\label{fig:delay_per}
\end{figure}
\begin{figure}[h]
\centering
\includegraphics[width=1\columnwidth]{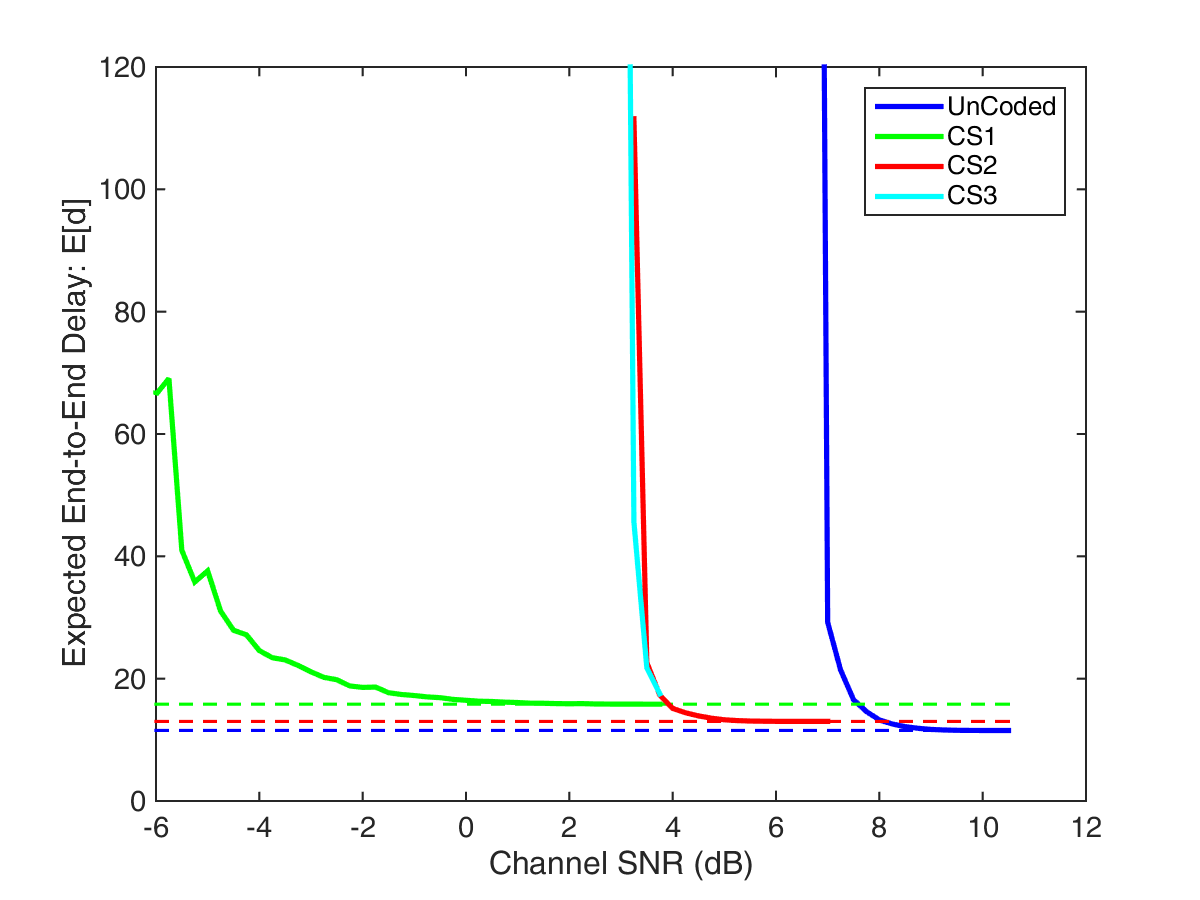}     
\caption{Expected delay $\mathbb{E}[d]$ vs Channel SNR for uncoded and coded system. In both CS1 and CS2, a RSC convolutional encoder with constraint length $7$ ans feedforward and feedback polynomials of $171$ hex and $133$ hex is used. In CS1, the coding rate for data and header is $R_D=R_H=1/2$, while in CS2, $R_D=R_H=3/4$ and in CS3, $R_D=3/4$ and $R_H=1/2$. Various coding rates are obtained using different puncturing patterns. The decoding scheme employs Viterbi algorithm. Other parameters are $\lambda=10, T=1.5, N=8, H=40, R=400$ bit/sec.}
\label{fig:EdCoded}
\end{figure}

Fig.~\ref{fig:delay_per} demonstrates the behavior of the expected average delay, $\mathbb{E}[\mathbf{d}]$ derived in (\ref{eq:Ed3}) with varying packetization interval $T$ for different channel qualities in terms of BER. It is shown that for small $T$, the high overhead cost may cause the average input rate (in terms of bit per second) exceed the service rate, hence the queue becomes unstable. The service rate apparently is lower for higher bit error probability that unstabilizes the queue for smaller input rates. On the other hand, for error-free channels, when $T$ grows to infinity, the packet inter-arrival times tend to be less than service time and therefore the dominant delay term is packet formation delay in (\ref{eq:f2}). Hence, the expected delay grows linearly with packetization interval length. For noisy channels with non-zero PER, longer packetization intervals produce a larger average packet lengths. Consequently, PER grows exponentially with $T$ and imposes larger service time, which in turn imposes longer waiting time on the queue that ultimately makes the queue unstable. Therefore, the growth of end-to-end delay with $T$ is more crucial for higher error probabilities. This shows the importance of the given analysis to find an optimal packetization interval, which is not only a function of input traffic properties, but also depends on the underlying channel quality. Therefore, ignoring physical layer parameters such as the utilized channel error-rate in the packet formation at higher layers of communications protocol might cause the whole transmission system to fail.

In order to investigate the effect of channel error probability in the proposed method, we evaluated the end-to-end delay for both uncoded and coded systems in Fig. \ref{fig:EdCoded} using the generalizations provided in section \ref{sec:sysmodel}. In this system, we use different coding rates for payload and header part of the produced packets similar to practical communication protocols \cite{IEEE802-11-1}. An immediate observation is that as the channel SNR improves, the end-to-end delay decreases, which is apparently due to a fewer number of re-transmissions. A more interesting observation is that an encoded system remains operational in a lower SNR range since the packet error rate (and consequent packet re-transmissions) remain low and hence the queue does not become unstable. The smaller the coding rate $R_D,R_H\rightarrow 0$, the longer header size is affordable before the queue becomes unstable. However, in high SNR regime, the uncoded system is desirable since it has shorter packet lengths for the same number of input symbols and consequently requires shorter service time. The coding scheme CS1 with code rates $R_H=1/2, R_D=1/2$  which uses more parity bits demonstrates a slightly lover minimum delay in high SNR regime compared to CS2. Delay floor effect is shown with slotted line in Fig. \ref{fig:EdCoded}.

\begin{figure}[t]
\centering
\includegraphics[width=1\columnwidth]{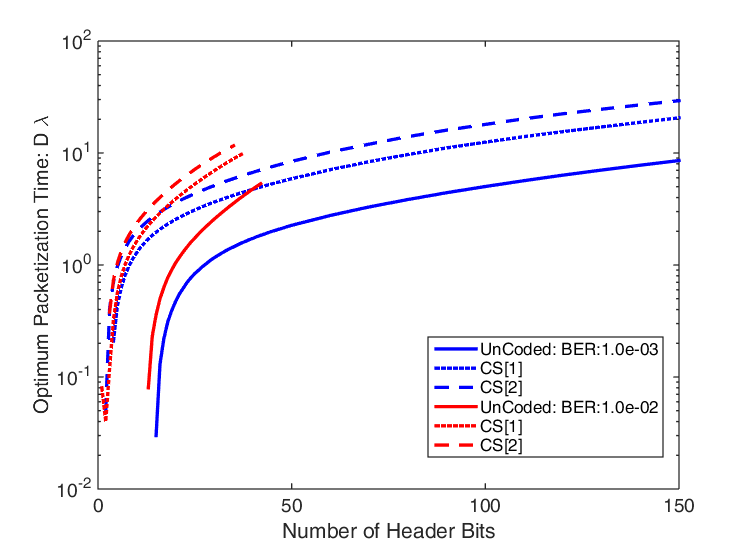}     
\caption{Optimal packetization interval vs header length $H$ for uncoded and coded systems with $N=8, \lambda = 10$, $R=400$ bit/sec. In coded scheme 1 (CS1), an encoder with rate $R_D=R_H=1/2$ and $\beta_D=\beta_H=\beta \text{(Uncoded)}/10$ is used. In coded scheme 2 (CS2), two encoders with rates $R_D=1/2$ and $R_H=1/3$  and corresponding bit error probability of $\beta_D=\beta/10$ and $\beta_H=\beta/100$ are used for payload data and header parts.}
\label{fig:delay2}
\end{figure}

Fig.~\ref{fig:delay2} presents the optimum packetization interval for different header sizes $H$ for both uncoded and coded systems. Obviously, a larger number of header bits demands higher packetization interval to compensate the waiting time caused by low header efficiency and intends to control the effective header bits per symbol $\frac{H}{\mu}$. For higher channel error, in the same header size, the optimum packetization interval is smaller to avoid a large number of re-transmissions. Another observation is that for a coded system with more parity bits and a better error rate (CS2 in Fig. \ref{fig:delay2}), the optimal packetization interval is set larger, since the number of retransmissions that may arise due to longer packet sizes are compensated by a better error correction property. Therefore, the system appropriately chooses an optimal packetization interval that minimizes the expected delay for both uncoded and coded systems. Thirdly, the system with higher error rate reaches (highlighted in red color in Fig.~\ref{fig:delay2}) the instability point with shorter header lengths due to more frequent packet retransmissions as expected.

\begin{figure}[h]
\centering
\includegraphics[width=1\columnwidth]{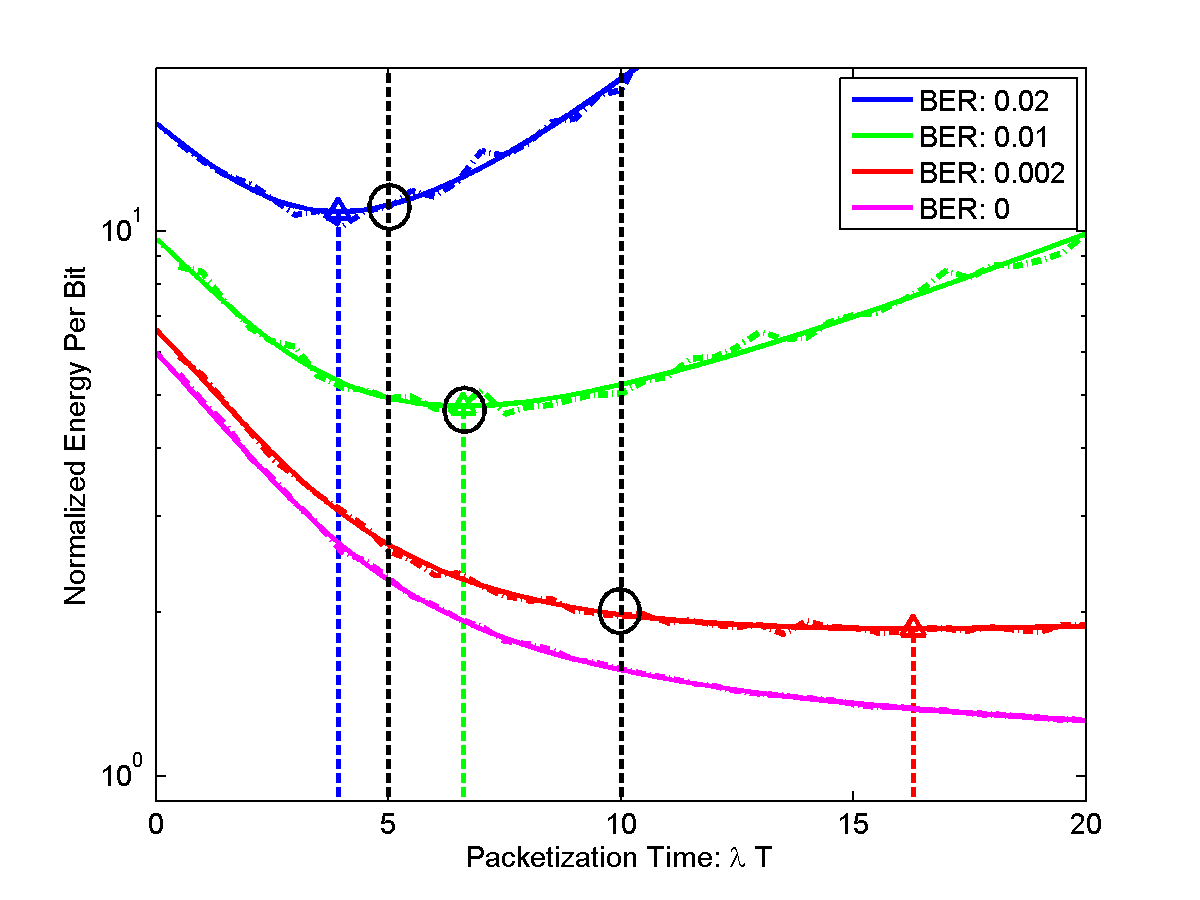}    
\caption{Energy Per Unit Bit vs Packetization Time $(N=8, H = 40, \lambda = 1)$.}
\label{fig:EF}
\end{figure}

In order to analyze the impact of packetization interval on Energy efficiency, the expected energy consumption per unit information bit (ECR) is depicted in Fig. \ref{fig:EF}. Due to the convexity of ECR(T) with respect to $T$, the most energy-efficient packetization interval $T^\ast$ can be found analytically or numerically evaluated using methods such as gradient descent. Energy efficiency is impacted by $T$ through two phenomena of header efficiency and number of re-transmissions. By increasing the packetization interval, the ratio of header bits to packet length ($\mathbb{E}[\frac{H}{H+\mathbf{k}N}]$) decreases, which in turn translates to a lower energy consumption per information bit. This reduction in average energy per bit continues until it is compensated by the extra average energy consumption required for re-transmission of longer packets, since the longer packets are more likely to be discarded. The global minimum corresponds to the balance between these two contradictory effects. It is observed that for error-free channels, the average energy consumption is decreasing function of $T$, since no re-transmission occurs. If one is interested in energy efficiency under constrained average delay ($\mathbb{E}[\mathbf{d}]<D_0$), the global minimum may fall out of the valid range of $T$ and hence one of the corner points (solutions of $\mathbb{E}[\mathbf{d}]=D_0$) provides the most energy efficient solution. For instance, if the average delay constraint implies $T \in [5/\lambda \sim 10/\lambda ]$ as depicted by dashed black color in this figure, the global minimum of energy curve for $\beta=0.01$ falls in the valid region and hence it is achievable. However, the packetization interval that minimizes the expected energy consumption falls out of the valid range of packetization interval for $\beta=0.02$ and $\beta=0.002$. Therefore, in such cases we need to choose one of the left and right corner points namely $T=5/\lambda$ and $T=10/\lambda$, which are imposed by delay constraint, as highlighted by circle in this figure.

\section *{Acknowledgment}
The authors would like to thank Prof. Anthony Ephremides and Dr. Jeongho Jeon for their thoughtful comments and discussions on the earlier version of this paper. 

\section {Concluding Remarks} \label{sec:conclusion}
In this paper, the impact of packetization interval length on the end-to-end latency is investigated to provide an optimal policy to encapsulate Poisson arrival symbols into transmit packets. Closed-form expressions are derived for both expected end-to-end delay and energy efficiency through queuing system analysis. It was noticed that the Poisson arrival symbols yield three distributions for packet inter-arrivals including: exponential, Geometric, and Deterministic as $T$ departs from zero to infinity.

It was also shown that a small packetization interval reduces header efficiency, which may cause the packet rates to exceed the service rate, unstabilize the queue and impose extremely large delay. On the other hand, a larger packetization interval increases the expected delay due to longer wait to form a packet as well as delay arose from more frequent re-transmissions. This suggests that packet formation policy design should incorporate underlying physical layer parameters such as channel rate and bit error probability in addition to network-based traffic statistics. The proposed policy not only optimizes the performance metrics such as delay and energy efficiency, but also avoids the transmission system crash due to queue instability. This delay-optimal packetization policy for a single hop communication is essential in sensor network entry points, where measurement symbols are bundled into packets. This fundamental study can be further extended to more complicated system setups in Ad-hoc networks and can also be integrated with delay-minimal scheduling policies to reduce the overall system delivery time.

\appendices

\appendix
\textbf{Proof of proposition \ref{prop:2}}:
In order to prove the proposition, we first state the following lemmas:
\begin{lemma} \label{lem:1}
If $\mathbf{k}$ is a Poisson r.v. with mean $\mu$, then $\mathbb{E}_\mathbf{k}[\frac{\mathbf{k}^n}{\zeta^\mathbf{k}}] =  e^{-\mu (1-1/ \zeta)} \sum_{i=1}^{n}S_2(n,i)  (\mu / \zeta)^i$ for $n=1,2,3...$, where $S_2(n,i)$ is the Stirling number of the second kind that counts the number of ways to partition a set of $n$ elements into $i$ nonempty subsets~\cite{Stirling1}.
\end{lemma}
\begin{proof}
First, we notice that:
\begin{align} \label{eq:kn1}
\nonumber
&\mathbb{E}_{\mathbf{k}\sim\mathbf{Poiss}(\mu)}[\frac{\mathbf{k}^n}{\zeta^\mathbf{k}}] = \sum_{k=0}^{\infty} \frac{k^n}{\zeta^k} e^{-\mu} \frac{\mu^k}{k!} = e^{-\mu (1-1/ \zeta)} \times
\\
 &\sum_{k=1}^{\infty} k^n e^{-\mu / \zeta} \frac{(\mu/ \zeta)^k}{k!} = e^{-\mu (1-1/ \zeta)} \mathbb{E}_{\mathbf{k\sim\mathbf{Poiss}(\mu/\zeta)}}[\mathbf{k}^n].
\end{align}
where $k\sim\mathbf{Poiss}(\mu)$ is to reflect that $k$ is a Poisson r.v. with mean($\mu$). It was shown in ~\cite{Stirling2} that $k^n$ can be represented in the form of falling factorials as 
\begin{align} \label{eq:str1}
k^n=\sum_{i=0}^{n} S_2(n,i)  k_{(i)},
\end{align}
where $k_{(i)}=k(k-1)...(k-n+1)$ is the falling factorial. 
Combining (\ref{eq:kn1}) and (\ref{eq:str1}) results in
\begin{align} \label{eq:kn1-str1}
\mathbb{E}_{\mathbf{k}\sim\mathbf{Poiss}(\mu)}[\frac{k^n}{\zeta^k}] &= e^{-\mu (1-1/ \zeta)} \sum_{i=0}^{n} S_2(n,i) \mathbb{E}_{\mathbf{k}\sim\mathbf{Poiss}(\mu/ \zeta)}[\mathbf{k}_{(i)}].
\end{align}
$\mathbb{E}_\mathbf{k}[\mathbf{k}_{(i)}]$ can be easily found by factorial moment generation function for a Poisson distribution as follows: 
\begin{align} \label{eq:kn2}
\nonumber \mathbb{E}_{\mathbf{k}(\mu/ \zeta)}[\mathbf{k}_{(i)}]&=\frac{d^i}{dt^i}\mathbb{E}_{\mathbf{k}(\mu/ \zeta)}[t^{\mathbf{k}}] \mid_{t=0}= \frac{d^i}{dt^i} e^{\mu/ \zeta(t-1)}\mid_{t=0}\\
& = \mu^i e^{\mu/ \zeta(t-1)}\mid_{t=0} = (\mu/ \zeta)^i
\end{align}
Substituting (\ref{eq:kn2}) in (\ref{eq:kn1-str1}) completes the proof.

\end{proof}

\textit{Corollary:} As special cases for $n=1,2$, we have $S_2(n,0)=0$ and $S_2(1,1)=S_2(2,1)=S_2(2,2)=1$. Therefore:
\begin{align} 
\label{eq:kn2-n1}
\mathbb{E}_\mathbf{k}[\frac{\mathbf{k}}{\zeta^\mathbf{k}}] &=  \frac{\mu}{\zeta} e^{-\mu (1-1/ \zeta)},   \\
\label{eq:kn2-n2}
\mathbb{E}_\mathbf{k}[\frac{\mathbf{k}^2}{\zeta^\mathbf{k}}] &=  \frac{\mu}{\zeta} (1+ \frac{\mu}{\zeta})e^{-\mu (1-1/ \zeta)}. 
\end{align}

\begin{lemma}  \label{lem:2}
If $\mathbf{k}$ is a Poisson r.v. with parameter $\mu$ and $x(\mathbf{k})$ is defined as $x(\mathbf{k})\triangleq \frac{\big(h(\mathbf{k})\big)^m}{\zeta^k}$ for $m=1,2,3,...$, then $\mathbb{E}_\mathbf{k}[x(k)]=  e^{-\mu} (e^{\mu/ \zeta}-1)$. 

The proof follows from definition of expected value. We notice that $\big(h(\mathbf{k})\big)^m=h(\mathbf{k})$. Thus,
\begin{align}
\nonumber 
&x(\mathbf{k})=\begin{cases} \frac{1}{\zeta^\mathbf{k}} & \mathbf{k} \neq 0\\0 &\mathbf{k}=0 \end{cases} \Longrightarrow 
\mathbb{E}_{\mathbf{k}}[x(\mathbf{k})]&=\sum_{k=0}^{\infty} x (k)p_\mathbf{k}(k)\\
\nonumber 
&=\sum_{k=0}^{\infty} \frac{1}{\zeta^k}p_\mathbf{k}(k)-p_\mathbf{k}(k=0)\frac{1}{\zeta^0}\\ 
&= \mathbb{E}_\mathbf{k}[\frac{1}{\zeta^k}]-e^{-\mu}=  e^{-\mu} (e^{\mu/ \zeta}-1). 
\end{align}
\end{lemma}

\begin{lemma}  \label{lem:3}
If $\mathbf{k}$ is a Poisson r.v. with parameter $\mu$ and $x(\mathbf{k})$ is defined as $x(\mathbf{k})\triangleq \frac{\big(h(\mathbf{k})\big)^m \mathbf{k}^n}{\zeta^\mathbf{k}},$ where $m,n \in \{1,2,3,...\}$, then $\mathbb{E}_\mathbf{k}[x(\mathbf{k})]=   e^{-\mu (1-1/ \zeta)} \sum_{i=1}^{n} {n \choose i} (\mu / \zeta)^i$. 

The proof immediately follows from lemma \ref{lem:1} and noting the fact that $x(\mathbf{k})= \frac{h(\mathbf{k})^m \mathbf{k}^n}{\zeta^\mathbf{k}} = \frac{\mathbf{k}^n}{\zeta^\mathbf{k}}$ for any positive integer $m$.
\end{lemma}

Now, we proceed with deriving the moments of the auxiliary service time $\tilde{\mathbf{s}}$, then calculate those of $\mathbf{s}$. We note that the re-transmission parameter, $\mathbf{r}$ is Geometrically distributed with success parameter $(1-\beta)^{l}=\alpha^{H+kN}$ which dependent on the packet length. Hence, its first and second order moments are $\mathbb{E}_\mathbf{r}[\mathbf{r}]=\alpha^{-(H+\mathbf{k}N)}$ and $\mathbb{E}_\mathbf{r}[\mathbf{r}^2]=\frac{2-\alpha^{H+\mathbf{k}N}}{\alpha^{2(H+\mathbf{k}N)}}$, which are functions of $\mathbf{k}$. Therefore, using Lemmas \ref{lem:1}, \ref{lem:2} and \ref{lem:3}, the moments of $\tilde{\mathbf{s}}$ can be calculated as
\begin{align}   \label{eq:Es1-proof}
\nonumber 
\mathbb{E}[\tilde{\mathbf{s}}] &= \mathbb{E}_\mathbf{k}\big[\mathbb{E}_\mathbf{r}[\tilde{\mathbf{s}}]\big] 
=\mathbb{E}_\mathbf{k}\big[\frac{h(\mathbf{k})H+kN}{R}\mathbb{E}_\mathbf{r}[\mathbf{r}]\big] \\
\nonumber
&=\mathbb{E}_\mathbf{k} \Big [\frac{h(\mathbf{k})H+\mathbf{k}N}{R\alpha^{H+\mathbf{k}N}} \Big] \\
\nonumber
&=\frac{H}{R\alpha^{H}}\mathbb{E}_\mathbf{k} \Big [\frac{h(\mathbf{k})}{(\alpha^{N})^\mathbf{k}} \Big] +
\frac{N}{R\alpha^{H}}\mathbb{E}_\mathbf{k} \Big [\frac{\mathbf{k}}{(\alpha^{N})^\mathbf{k}} \Big] \\
\nonumber
&=\frac{H}{R\alpha^{H}}  \big ( e^{-\mu}(e^{\mu \alpha^{-N}} -1) \big)   + \frac{\mu N}{R\alpha^{H}\alpha^N} e^{-\mu (1-\alpha^{-N}) } \Big]\\
&=\frac{Ne^{-\mu}}{R\alpha^{H}} \Big [  (\eta + \mu{\alpha}^{-N}) e^{ \mu \alpha^{-N} } -\eta  \Big],
\end{align}
where $\eta=H/N$ denotes the ratio of header size to symbol size.
Similarly we have
\begin{align}  \label{eq:Es2-proof}
\nonumber
\mathbb{E}[\tilde{\mathbf{s}}^2] &= \mathbb{E}_\mathbf{k}[\mathbb{E}_\mathbf{r}(\tilde{\mathbf{s}}^2)] 
=\mathbb{E}_\mathbf{k}\big[\frac{(h(\mathbf{k})H+\mathbf{k}N)^2}{R^2}\mathbb{E}[\mathbf{r}^2]\big] 
\\ \nonumber
& =\mathbb{E}_\mathbf{k} \Big [\frac{(h(\mathbf{k})H+\mathbf{k}N)^2(2-\alpha^{H+\mathbf{k}N})}{R^2\alpha^{2(H+\mathbf{k}N)}} \Big]\\
\nonumber
&=\frac{2H^2}{R^2\alpha^{2H}}\mathbb{E}_\mathbf{k} \Big [\frac{h^2(\mathbf{k})}{(\alpha^{2N})^\mathbf{k}} \Big] 
+\frac{2N^2}{R^2\alpha^{2H}}\mathbb{E}_\mathbf{\mathbf{k}} \Big [\frac{\mathbf{k}^2}{(\alpha^{2N})^\mathbf{k}} \Big]\\
\nonumber
& +\frac{4NH}{R^2\alpha^{2H}}\mathbb{E}_\mathbf{k} \Big [\frac{\mathbf{k} h(\mathbf{k})}{(\alpha^{2N})^\mathbf{k}} \Big]  -\frac{H^2}{R^2\alpha^{H}}\mathbb{E}_\mathbf{k} \Big [\frac{h^2(\mathbf{k})}{(\alpha^{N})^\mathbf{k}} \Big] \\
& -\frac{N^2}{R^2\alpha^{H}}\mathbb{E}_\mathbf{k} \Big [\frac{\mathbf{k} ^2}{(\alpha^{N})^\mathbf{k}} \Big] 
-\frac{2NH}{R^2\alpha^{H}}\mathbb{E}_\mathbf{k} \Big [\frac{\mathbf{k}h(\mathbf{k})}{(\alpha^{N})^\mathbf{k}} \Big]
\end{align}
After substituting the relevant expected values noting Lemma \ref{lem:1} and \ref{lem:2} and some simple manipulations, we obtain:
\begin{align}  \label{eq:Es2-middle0}
\nonumber
\mathbb{E}[\tilde{\mathbf{s}}^2] &=\frac{2H^2}{R^2\alpha^{2H}}\Big [   e^{-\mu}(e^{\mu \alpha^{2N}}-1) \Big] \\
\nonumber
&+\frac{2N^2}{R^2\alpha^{2H}}\Big [e^{-\mu(1-\alpha^{2N})} (2 \mu \alpha^{2N} + \mu^2 \alpha^{2N})  \Big] \\
\nonumber
&+\frac{4NH}{R^2\alpha^{2H}}\Big [e^{-\mu(1-\alpha^{2N})}(\mu \alpha^{2N}) \Big] -\frac{H^2}{R^2\alpha^{H}}\Big [e^{-\mu}(e^{\mu \alpha^{N}}-1)  \Big] \\
\nonumber
& -\frac{N^2}{R^2\alpha^{H}}\Big [e^{-\mu(1-\alpha^{N})} (2 \mu \alpha^{N} + \mu^2 \alpha^{2N})  \Big] \\
&-\frac{2NH}{R^2\alpha^{H}}\Big [  e^{-\mu(1-\alpha^{N})}(\mu \alpha^{N})\Big]
\end{align}

\begin{align} \label{eq:Es2-proof-2}
\nonumber
\mathbb{E}[\tilde{\mathbf{s}}^2]&=\frac{N^2 e^{-\mu}}{R^2 \alpha^H}  \Big [ \big ( 2 \mu \alpha^{-2N} + 2 \lambda^2 T^2 \alpha^{-4N} + 4\eta \mu \alpha^{-2N} + 2\eta^2 \big)\\
\nonumber
&. \alpha^{-H}  e^{\mu \alpha^{-2N}} 
-\big (  \mu \alpha^{-N} +  \lambda^2 T^2 \alpha^{-2N} 
+ 2\eta \mu \alpha^{-N} + \eta^2 \big) \\
&.e^{\mu \alpha^{-N}} - \big(2\eta^2\alpha^{-H}-\eta^2\big)  \Big ].
\end{align}
Since no packet is generated for zero symbol accumulation during an interval, we exclude the packets with zero symbols and length $H$ according to (\ref{eq:ss}) to obtain the moments for service time as follows:  
\begin{align} \label{eq:Es1-2}
\nonumber
&\mathbb{E}[\tilde{\mathbf{s}}^n]=\mathbb{E}[\tilde{\mathbf{s}}^n|\mathbf{k}=0]\mathbb{P}(\mathbf{k}=0)+\mathbb{E}[\tilde{\mathbf{s}}^n|\mathbf{k} \neq 0]\mathbb{P}(\mathbf{k} \neq0)\\
&=\mathbb{E}[\mathbf{s}^n]\mathbb{P}(\mathbf{k} \neq0) \Longrightarrow \mathbb{E}[\mathbf{s}^n]= \frac{\mathbb{E}[\tilde{\mathbf{s}}^n]}{1-e^{-\mu}}.
\end{align}

Substituting (\ref{eq:Es1-proof}) and (\ref{eq:Es2-proof}) in (\ref{eq:Es1-2}) completes the proof.\\

\noindent
\textbf{Proof of lemma \ref{lem:Ef}}:
In order to proof (\ref{eq:Eflemma}), we define $f_{\xi}(\mu)=E_\mathbf{k}[\frac{1}{\mathbf{k} \xi^\mathbf{k}}]$. Noting $\mathbb{P}(\mathbf{k}=k)=\frac{1}{1-e^{-\mu}}e^{-\mu} \mu^k/k!$ we have
\begin{align} \label{eq:Eflemma-proof}
\nonumber
f_\xi(\mu)&=\mathbb{E}[\frac{1}{\mathbf{k} \xi^\mathbf{k}}]=\sum\limits_{k=1}^{\infty} \frac{1}{k\xi^k} 
\frac{e^{-\mu}\mu^k}{(1-e^{-\mu})k!}\\
&= \frac{e^{\mu/\xi-\mu}}{(1-e^{-\mu})} \sum\limits_{k=1}^{\infty} \frac{(\mu/\xi)^k e^{-\mu/eta}}{k. k!} 
\end{align}

For notation convenience, we define $g_\xi(x)$ with change of variables as follows:
\begin{align} \label{eq:Eflemma-proof2}
g_\xi(x) = (1-e^{-\mu})e^{\mu-\mu/\xi} f_\xi(\mu)|_{x=\mu/\xi}= \sum\limits_{k=1}^{\infty} \frac{x^k e^{-x}}{k. k!} 
\end{align}
Now, we take derivative of (\ref{eq:Eflemma-proof2}) with respect to $x$ to obtain the following ordinary differential equation (ODE):
\begin{align} \label{eq:Eflemma-proof3}
\nonumber
\frac{dg_{\xi}(x)}{dx}&=\sum\limits_{k=1}^{\infty} \frac{1}{k. k!} \frac{d}{dx}(x^k e^{-x})=\sum\limits_{k=1}^{\infty} \frac{1}{k. k!} (kx^{k-1} -x^k e^{-x})\\
\nonumber
&= \frac{1}{x}\sum\limits_{k=1}^{\infty} \frac{1}{k!} x^{k}e^{-x} - g_\xi(x)= \frac{1-e^{-x}}{x} - g_\xi(x)\\
\Longleftrightarrow& \frac{dg_{\xi}(x)}{dx}+g_\xi(x)= \frac{1-e^{-x}}{x}
\end{align}
The ODE in (\ref{eq:Eflemma-proof3}) can be solved as follows:
\begin{align} \label{eq:Eflemma-proof4}
\nonumber
& \frac{dg_{\xi}(x)}{dx}+g_\xi(x)= \frac{1-e^{-x}}{x}\\
\nonumber
&\Longleftrightarrow \frac{dg_{\xi}(x)}{dx} e^x+g_\xi(x) e^x= \frac{e^x-1}{x}\\
\nonumber
&\Longleftrightarrow \frac{d}{dx}(g_{\xi}(x)e^x)= \frac{e^x-1}{x}\\
\nonumber
&\Longleftrightarrow g_{\xi}(x)e^x= \int_{\infty}^{x} t^{-1} e^t dt - \int_{\infty}^{x} t^{-1} dt = \mathcal{E}(x)-\log(x) + c\\
&\Longleftrightarrow g_{\xi}(x)= e^{-x}\big(\mathcal{E}(x)-\log(x)+c\big)
\end{align}
where $\log(.)$ is the natural logarithm and $\mathcal{E}(x)=\int_{-\infty}^{x} \frac{e^{t}}{t}\text{dt}$ is called the exponential integral. Noting $\lim_{x \to 0}g_{\xi}(x) = 0$, the constant $c$ is determined to be negative of Euler constant, $c=-\gamma$.  The value of $\mathcal{E}(x)$ can be readily evaluated using available tables, recursive algorithms and Taylor series \cite{Euler2013}.
Now, using we substitute (\ref{eq:Eflemma-proof4}) in (\ref{eq:Eflemma-proof2}) to obtain the result:
\begin{align} \label{eq:Eflemma-proof2}
f_\xi(\mu) = \frac{e^{\mu/\xi-\mu}}{(1-e^{-\mu})}g_\xi(x)  |_{\mu=x \xi}=\frac{\mathcal{E}(\mu/\eta)-\log(\mu/\eta)-\gamma}{e^\mu-1}
\end{align}

\bibliographystyle{IEEEtran}	
\bibliography{BibsPD1}

\end{document}